\documentclass[11pt,a4paper]{article}
\usepackage{amsfonts}
\usepackage[dvips]{graphicx}
\usepackage{amsmath}
\usepackage{mathrsfs}
\usepackage{makeidx}
\usepackage{xypic}
\usepackage[nottoc]{tocbibind}
\usepackage{pstricks}
\usepackage{bbm}
\usepackage[arrow, matrix, curve]{xy}
\usepackage{amsthm}
\usepackage{amssymb}
\usepackage{epic}
\usepackage{eepic}
\usepackage{times}
\usepackage{epsfig}

\xymatrixcolsep{0.5cm}
\xymatrixrowsep{0.5cm}
\newtheorem{theorem}{Theorem}[section]
\newtheorem{proposition}{Proposition}[section]
\newtheorem{lemma}{Lemma}[section]

\newtheorem{corollary}{Corollary}[section]

\newcommand{\beqa}{\begin{eqnarray}}
\newcommand{\eeqa}{\end{eqnarray}}
\newcommand{\rf}[1]{(\ref{#1})}

\makeindex
\numberwithin{equation}{section}

\begin{document}
\begin{flushright}
YITP-SB-12-
\end{flushright}

\bigskip

\begin{center}
\textbf{\Large Antiperiodic spin-1/2 $XXZ$ quantum chains by separation of
variables:}\vspace{2pt}

\textbf{\Large \ Complete spectrum and form factors}

\vspace{45pt}

{\large G.~Niccoli}\footnote{%
YITP, Stony Brook University, New York 11794-3840, USA,
niccoli@max2.physics.sunysb.edu}{\large \ ,}

\vspace{50pt}

\vspace{80pt}
\end{center}

\begin{small}
\begin{itemize}
\item[] \textbf{Abstract}\,\,\,In this paper we consider the spin 1/2 highest
weight representations for the 6-vertex Yang-Baxter algebra on a finite
lattice and analyze the integrable quantum models associated to the
antiperiodic transfer matrix. For these models, which in
the homogeneous limit reproduces the $XXZ$ spin 1/2 quantum chains with
antiperiodic boundary conditions, we obtain in the framework of
Sklyanin's quantum separation of variables (SOV) the following results: \textsf{I)} The complete
characterization of the transfer matrix spectrum (eigenvalues/eigenstates)
and the proof of its simplicity. \textsf{II)} The reconstruction of all local
operators in terms of Sklyanin's quantum separate variables. \textsf{III)%
} One determinant formula for the scalar products of \textit{separates states%
}, the elements of the matrix in the scalar product are sums over the
SOV spectrum of the product of the coefficients of the states. \textsf{IV)} The form
factors of the local spin operators on the transfer matrix eigenstates by a
one determinant formula given by simple modifications of the scalar product
formula.
\end{itemize}
\end{small}

\newpage

\tableofcontents

\newpage

\section{Introduction}
The exact and complete solution of the integrable quantum models by the computation of their spectrum and dynamics is a
central issue for the mathematical physics as they play an important role in
different research areas. The general correspondence existing between these 1-dimensional
quantum models\footnote{%
See \cite{ID91,McCoy94} and references therein.} and 2-dimensional
(exactly solvable) models of classical statistical mechanics\footnote{%
See \cite{BaxBook} and reference therein.} gives an important example of
this statement. Another important application is to the quantum statistical
systems at finite temperature where the integrable Hamiltonians are used to
define partition functions and quantum thermal averages\footnote{%
That is the thermodynamical Bethe ansatz \cite{DMB69}-\cite{Za90}.}. Here, we present an approach to achieve the exact and complete solution of lattice integrable quantum models in the framework of the quantum
inverse scattering method (QISM) \cite{SF78}-\cite{IK82}. In particular, this approach is
addressed to the large class of integrable quantum models whose spectrum\footnote{
In the QISM formulation the quantum integrable structure (i.e. the complete
set of commuting conserved charges) of the model is generated by the
transfer matrix, the trace of the monodromy matrix satisfying the
Yang-Baxter algebra or generalization of it.} (eigenvalues \& eigenstates) can be determined by implementing
Sklyanin's quantum separation of variables (SOV) method \cite{Sk1}-\cite{Sk3}%
. Let us comment that the approach here
presented can be considered as the generalization to the SOV framework of
the Lyon group method\footnote{Always in the ABA framework, see also \cite{K01}-\cite{CM07} for the
extension of this method to the higher spin quantum chains and \cite{KKMNST07}-\cite{KKMNST08} for the generalization to the reflection algebra case.} \cite{KitMT99}-%
\cite{MaiT00} which was instead developed in the algebraic Bethe ansatz
(ABA) framework \cite{SF78}-\cite{FST80}.

This approach has been recently developed in the case of the lattice
quantum sine-Gordon model \cite{FST80,IK82}\ and the $\tau _{2}$-model \cite%
{Ba04} associated by QISM to cyclic representations of the 6-vertex
Yang-Baxter algebra. In particular, in the papers \cite{NT-10}-\cite{N-11}
the complete SOV spectrum characterization has been constructed for the
lattice quantum sine-Gordon model while in the paper\footnote{See also the series of works \cite{GIPS06}-\cite{GIPS09} for previous
analysis by SOV method of the $\tau _{2}$-model and for some first result
concerning the computation of the form factors of local operators in the
special case of the generalized Ising model.} \cite{GN12} it has been
derived for the $\tau _{2}$-model and consequently for the chiral Potts
model \cite{BS90}-\cite{TarasovSChP}, by exploiting the well known links
between these two models \cite{BS90}. Finally, in the papers \cite{GMN12-SG,GMN12-T2} it has been shown how to reconstruct local operators in terms of the quantum separate variables and write in a determinant form the scalar products of separate states\footnote{See Section \ref{SP} for the definition of these states in our current model.} and the matrix elements of local operators on transfer matrix eigenstates. 

In the present article we develop this approach for quantum models associated by QISM to highest weight representations of the 6-vertex
Yang-Baxter algebra. In particular, we consider the representations
corresponding to one of the most prototypical lattice integrable quantum
model, i.e. the $XXZ$ spin 1/2 quantum chain \cite{H28}. It represents one of the
best known quantum models under periodic boundary conditions and a very
large literature is dedicated to it \cite{Be31}-\cite{LM66}. In particular,
it was the basic example for the application of the ABA method and then for
the implementation of the Lyon group method to compute matrix elements of
local operators.

The circumstance \cite{Sk2} interesting for us is that it is enough to
change the boundary conditions into antiperiodic ones that the algebraic
Bethe ansatz does not work anymore while Sklyanin's quantum separation
of variables can be used to analyze the system. Moreover, it is worth remarking that the thermodynamical limits of the
periodic and antiperiodic case are naturally expected coinciding and then the $XXZ$ quantum spin 1/2 chain with
antiperiodic boundary conditions is a natural prototype for which develop
the full program of analysis from the spectrum up to the correlation
functions in the SOV framework. As for this model we can take advantage from the known results worked out for the
periodic chain in the ABA framework by the Lyon group \cite{IKitMT99}-\cite{KKMST07} and  compare them with our findings.

Let us mention that so far the results concerning the antiperiodic case are restricted to the construction of the Baxter Q-operator\footnote{%
See also \cite{YB95} for the construction of the Q-operator in the higher
spin $XXZ$ quantum chain with twisted boundary conditions.} \cite{BBOY95} and
that of the functional separation of variables of Sklyanin for the $XXX$ case 
\cite{Sk2} extended in \cite{NWF09} to the $XXZ$ case. It is worth pointing out that Sklyanin's separation of variables in its
functional version defines representations of the Yang-Baxter algebra on
space of symmetric functions and leads only to the representation of the
wave functions of the transfer matrix eigenstates. In fact, the explicit construction of the SOV representation as well as
of the transfer matrix eigenstates in the original representation space of the quantum spin chain are missing. These are important information in order to
achieve the goal to compute matrix elements of the local operators and one
of the tasks of the present article is to cover these gaps. 

\subsection{Motivation for the use of SOV method}

In the framework of quantum integrability, there are several methods to
analyze the spectral problem as for example the coordinate Bethe ansatz \cite%
{Be31}, \cite{BaxBook} and \cite{ABBBQ87}, the Baxter Q-operator method \cite%
{BaxBook}, the algebraic Bethe ansatz \cite{SF78}-\cite{FST80}, the analytic
Bethe ansatz \cite{Re83-1}-\cite{Re83-2}. However, they suffer in general
from one or more of the following problems: i) Reduced applicability; i.e.
there exist important examples of quantum integrable models to which some of these
methods do not apply. ii) Analysis reduced only to the set of eigenvalues;
i.e. some of them do not allow for the construction of the eigenstates. iii) Lack of
completeness proof; i.e. the completeness of the spectrum description is not
assured by the methods but has to be separately proven\footnote{%
Note that there are only a few examples of integrable quantum models where
the completeness has been proven in the ABA framework, including the $XXX$
Heisenberg model; see \cite{MTV09} and references therein.}.

The SOV method of Sklyanin is a more promising approach: It works for a
large class of integrable quantum models to which ABA does not apply; it
leads to both the eigenvalues and the eigenstates of the transfer matrix
with a spectrum construction (which under simple conditions) has as built-in
feature its completeness. Moreover, for the so far analyzed cases \cite%
{NT-10}-\cite{N-11}, \cite{GN12}, \cite{N12-1}-\cite{N12-3} in the SOV
framework it was an easy task to prove the simplicity of the transfer matrix
or to add to it commuting operators which form a complete\footnote{%
The completeness (i.e. the non-degeneracy of spectrum) of the set of
commuting conserved charges is a natural requirement to state the complete
integrability of the quantum model as it represents the natural quantum
analogue to the classical definition of complete integrability which
requires the existence of a maximal number of independent and mutually in
involution integrals of motions.} set of commuting conserved charges of the
quantum model.

\section{Antiperiodic 6-vertex quantum integrable chain}

\subsection{Representations on spin-1/2 chain of Yang-Baxter algebra}

Let us define a class of representations of the 6-vertex Yang-Baxter algebra
on spin-1/2 quantum chains. More in details, let us denote with $\sigma
_{n}^{\pm }$ and $\sigma _{n}^{z}$ the generators of $\mathsf{N}$ independent
(local) $sl(2)$ algebras:
\begin{equation}
\lbrack \sigma _{n}^{z},\sigma _{m}^{\pm }]=\pm \delta _{n,m}\sigma
_{n}^{\pm },\text{ \ }[\sigma _{n}^{+},\sigma _{m}^{-}]=2\delta _{n,m}\sigma
_{n}^{z},
\end{equation}%
and let us introduce the 2-dimensional linear spaces (local quantum spaces
of the chain) R$_{n}\simeq $ $\mathbb{C}^{2}$. In any linear space R$_{n}$
is define a spin-1/2 representation of the $sl(2)$ algebra where the
generators of the algebra admit the standard representation in terms of $%
2\times 2$ Pauli matrices. Then, for any local quantum space R$_{n}$ with $%
n\in \{1,...,\mathsf{N}\}$, we can define the so-called Lax operator $%
\mathsf{L}_{0n}(\lambda )\in $ End$($R$_{0}\otimes $R$_{n})$: 
\begin{equation}
\mathsf{L}_{0n}(\lambda )\equiv \left( 
\begin{array}{cc}
A_{n}(\lambda ) & B_{n} \\ 
C_{n} & D_{n}(\lambda )%
\end{array}%
\right) _{0}=\left( 
\begin{array}{cc}
x_{+}(\lambda )+x_{-}(\lambda )\sigma _{n}^{z} & (q-q^{-1})\sigma _{n}^{-},
\\ 
(q-q^{-1})\sigma _{n}^{+} & x_{-}(\lambda )+x_{+}(\lambda )\sigma _{n}^{z}%
\end{array}%
\right) _{0},
\end{equation}%
where we have denoted $q=e^{\eta }\in \mathbb{C}$ and%
\begin{equation}
x_{\pm }(\lambda )\equiv (\lambda q-(q\lambda )^{-1}\pm \lambda -\lambda
^{-1})/2.
\end{equation}%
$\mathsf{L}_{0n}(\lambda )$ is a solution of the Yang-Baxter equation:%
\begin{equation}
R_{12}(\lambda /\mu )\mathsf{L}_{1n}(\lambda )\mathsf{L}%
_{2n}(\mu )=\mathsf{L}_{2n}(\mu )\mathsf{L}%
_{1n}(\lambda )R_{12}(\lambda /\mu ),
\end{equation}%
w.r.t. the 6-vertex R-matrix: 
\begin{equation}
R_{12}(\lambda )\equiv \left( 
\begin{array}{cccc}
\lambda q-(q\lambda )^{-1} & 0 & 0 & 0 \\ 
0 & \lambda -\lambda ^{-1} & q-q^{-1} & 0 \\ 
0 & q-q^{-1} & \lambda -\lambda ^{-1} & 0 \\ 
0 & 0 & 0 & \lambda q-(q\lambda )^{-1}%
\end{array}%
\right).
\end{equation}%
Then, we can introduce the so-called monodromy matrix:%
\begin{equation}
\mathsf{M}_{0}(\lambda )\equiv \left( 
\begin{array}{cc}
\mathsf{A}(\lambda ) & \mathsf{B}(\lambda ) \\ 
\mathsf{C}(\lambda ) & \mathsf{D}(\lambda )%
\end{array}%
\right) \equiv \mathsf{L}_{0\mathsf{N}}(\lambda _{\mathsf{N}})\cdots \mathsf{%
L}_{01}(\lambda _{1})\in \text{End}(\text{R}_{0}\otimes \mathcal{R}_{\mathsf{%
N}}),
\end{equation}%
where $\mathcal{R}_{\mathsf{N}}\equiv \otimes _{i=1}^{\mathsf{N}}$R$_{n}$,$\
\lambda _{n}\equiv \lambda /\eta _{n}$ and the $\eta _{n}\in \mathbb{C}$ are
called inhomogeneity parameters. The monodromy matrix $\mathsf{M}(\lambda )$
is also a solution of the Yang-Baxter equation: 
\begin{equation}
R_{12}(\lambda /\mu )\mathsf{M}_{1}(\lambda )\mathsf{M}_{2}(\mu )=\mathsf{M}%
_{2}(\mu )\mathsf{M}_{1}(\lambda )R_{12}(\lambda /\mu )\,,  \label{YBE0}
\end{equation}%
and its elements $\mathsf{A},$ $\mathsf{B},$ $\mathsf{C}$ and $\mathsf{D}$ are the generators of a 2$^{\mathsf{N}}$-dimensional representation of
Yang-Baxter algebra on $\mathcal{R}_{\mathsf{N}}$.

\subsubsection{Left and right representations of Yang-Baxter algebra}\label{v-scalar-product}

Let us denote with $|k,n\rangle $ the standard spin basis for the $2$%
-dimensional linear space R$_{n}$:%
\begin{equation}
\sigma _{n}^{z}|k,n\rangle =k|k,n\rangle ,\text{ \ }k\in \{-1,1\},
\end{equation}%
i.e. the $\sigma _{n}^{z}$-eigenbasis of the local space R$_{n}$. Let L$_{n}$
be the linear space dual of R$_{n}$ and let $\langle k,n|$ be the elements
of the dual spin basis defined by:%
\begin{equation}
\langle k,n|k^{\prime },n\rangle =(|k,n\rangle ,|k^{\prime },n\rangle
)\equiv \delta _{k,k^{\prime }}\text{ \ \ }\forall k,k^{\prime }\in \{-1,1\},
\end{equation}%
i.e. the covectors $\langle k,n|$ define the $\sigma _{n}^{z}$-eigenbasis in
the dual linear space L$_{n}$. In the \textit{left} (covectors) and \textit{%
right} (vectors) linear spaces:%
\begin{equation}
\mathcal{L}_{\mathsf{N}}\equiv \otimes _{n=1}^{\mathsf{N}}\text{L}_{n},\text{
\ \ \ \ }\mathcal{R}_{\mathsf{N}}\equiv \otimes _{n=1}^{\mathsf{N}}\text{R}%
_{n},
\end{equation}%
the representations of the local $sl(2)$ generators induce left and right
representations of dimension $2^{\mathsf{N}}$ of the monodromy matrix
elements, i.e. 2$^{\mathsf{N}}$-dimensional representations of Yang-Baxter
algebra with $\mathsf{N}$ parameters the inhomogeneities.

\subsubsection{Antiperiodic transfer matrix $\mathsf{\bar{T}}(\protect%
\lambda )$}

The Yang-Baxter equations (\ref{YBE0}) and the commutation relations:%
\begin{equation}
\lbrack R_{12}(\lambda /\mu ),\Sigma _{1}^{(\alpha ,b)}\otimes \Sigma
_{2}^{(\alpha ,b)}]=0,
\end{equation}%
where:%
\begin{equation}
\Sigma _{0}^{(\alpha ,b)}=\left( \sigma _{0}^{x}\right) ^{b}\left( 
\begin{array}{cc}
e^{\alpha } & 0 \\ 
0 & e^{-\alpha }%
\end{array}%
\right) _{0}\text{ \ \ \ \ \ }\forall \alpha \in \mathbb{C},\text{ }b=0,1,
\end{equation}%
imply that the transfer matrix:%
\begin{equation}
\mathsf{T}^{(\alpha ,b)}(\lambda )=\text{tr}_{0}[\Sigma _{0}^{(\alpha ,b)}%
\mathsf{M}_{0}(\lambda )],
\end{equation}%
for any fixed $\Sigma _{0}^{(\alpha ,b)}$, generates a one-parameter family
of commuting operators on $\mathcal{R}_{\mathsf{N}}$. Let us recall that the
so-called quantum determinant: 
\begin{equation}
\det \mathsf{M}(\lambda )\,\equiv \,\mathsf{A}(\lambda )\mathsf{D}(\lambda
/q)-\mathsf{B}(\lambda )\mathsf{C}(\lambda /q),  \label{q-det-f}
\end{equation}%
is a central element\footnote{%
The centrality of the quantum determinant in the Yang-Baxter algebra was
first discovered in \cite{IK81}; see also \cite{IK09} for an historical note.%
} of the Yang-Baxter algebra (\ref{YBE0}) which admits the following
factorized form: 
\begin{equation}
\det \mathsf{M}(\lambda )=\prod_{n=1}^{\mathsf{N}}\text{det}_{\text{q}%
}\mathsf{L}_{0n}(\lambda _{n}),
\end{equation}%
in terms of the local quantum determinants: 
\begin{equation}
\det \mathsf{L}_{0n}(\lambda )\equiv A_{n}(\lambda )D_{n}(\lambda
/q)-B_{n}C_{n},
\end{equation}%
which explicitly reads:%
\begin{equation}
\det \mathsf{M}(\lambda )\equiv -a(\lambda )d(\lambda /q),\text{ \ \ }%
a(\lambda )\equiv -\prod_{n=1}^{\mathsf{N}}(\frac{\lambda q}{\eta _{n}}-%
\frac{\eta _{n}}{\lambda q}),\text{ \ \ \ \ }d(\lambda )\equiv \prod_{n=1}^{%
\mathsf{N}}(\frac{\lambda }{\eta _{n}}-\frac{\eta _{n}}{\lambda }).
\end{equation}%
In the following we will analyze the spectral problem for the antiperiodic
transfer matrix:%
\begin{equation}
\mathsf{\bar{T}}(\lambda )\equiv \mathsf{B}(\lambda )+\mathsf{C}(\lambda )=%
\mathsf{T}^{(\alpha =0,b=1)}(\lambda ),
\end{equation}%
then it is important to point out the conditions under which this transfer
matrix is normal:

\begin{lemma}
$\mathsf{I)}$ In the massless regime, i.e. for $q=e^{\eta }$ a pure phase ($%
\eta \in i\mathbb{R}$), when all the inhomogeneities \{$\eta _{1},...,\eta _{%
\mathsf{N}}$\} are real numbers then the transfer matrix $\mathsf{\bar{T}}%
(\lambda )$ is a one parameter family of normal operators and the family:%
\begin{equation}
i\mathsf{\bar{T}}(\lambda )
\end{equation}%
is self-adjoint for any $\lambda q^{1/2}\in \mathbb{R}$.

$\mathsf{II)}$  In the massive regime, i.e. $q=e^{\eta }\in \mathbb{R}^{+}$ ($%
\eta \in \mathbb{R}$), when all the inhomogeneities \{$\eta _{1},...,\eta _{%
\mathsf{N}}$\} are pure phases then the transfer matrix $\mathsf{\bar{T}}%
(\lambda )$ is a one parameter family of normal operators and the family:%
\begin{equation}
i^{\mathsf{e}_{\mathsf{N}}}\mathsf{\bar{T}}(\lambda ),
\end{equation}%
\ \ \ where $\mathsf{e}_{\mathsf{N}}=\{1$ \ for $\mathsf{N}$ even, $0$ for 
$\mathsf{N}$ odd$\}$, is self-adjoint for any $\lambda q^{1/2}$ a pure
phase.
\end{lemma}

\begin{proof}
In the case \textsf{I)} it is trivial to verify that the local Lax operators 
$\mathsf{L}_{0n}(\lambda )$ satisfies the following Hermitian conjugation property:%
\begin{equation}
\mathsf{L}_{0n}(\lambda )^{\dagger }\equiv \sigma _{0}^{y}\mathsf{L}_{0n}(\lambda ^{\ast
}/q)\sigma _{0}^{y},
\end{equation}%
where $\dagger $ means the complex conjugation plus the transposition w.r.t.
the local quantum space $n$. Then, for the monodromy matrix it
follows:%
\begin{equation}
\mathsf{M}(\lambda )^{\dagger }\equiv \left( 
\begin{array}{cc}
\mathsf{A}^{\dagger }(\lambda ) & \mathsf{B}^{\dagger }(\lambda ) \\ 
\mathsf{C}^{\dagger }(\lambda ) & \mathsf{D}^{\dagger }(\lambda )%
\end{array}%
\right) =\left( 
\begin{array}{cc}
\mathsf{D}(\lambda ^{\ast }/q) & -\mathsf{C}(\lambda ^{\ast }/q) \\ 
-\mathsf{B}(\lambda ^{\ast }/q) & \mathsf{A}(\lambda ^{\ast }/q)%
\end{array}%
\right) ,
\end{equation}%
when all the inhomogeneities \{$\eta _{1},...,\eta _{\mathsf{N}}$\} are real
numbers. Then the transfer matrix $\bar{\mathsf{T}}(\lambda )$ is normal for
any $\lambda \in \mathbb{C}$ and the statement in \textsf{I)} simply follows.

In the case \textsf{II)} it holds:%
\begin{equation}
\mathsf{L}_{0n}(\lambda )^{\dagger }\equiv \sigma _{0}^{x}\mathsf{L}_{0n}(-1/(\lambda ^{\ast
}q))\sigma _{0}^{x}.
\end{equation}%
Then, for the monodromy matrix it follows:%
\begin{equation}
\mathsf{M}(\lambda )^{\dagger }\equiv \left( 
\begin{array}{cc}
\mathsf{A}^{\dagger }(\lambda ) & \mathsf{B}^{\dagger }(\lambda ) \\ 
\mathsf{C}^{\dagger }(\lambda ) & \mathsf{D}^{\dagger }(\lambda )%
\end{array}%
\right) =\left( 
\begin{array}{cc}
\mathsf{D}(-1/(\lambda ^{\ast }q)) & \mathsf{C}(-1/(\lambda ^{\ast }q)) \\ 
\mathsf{B}(-1/(\lambda ^{\ast }q)) & \mathsf{A}(-1/(\lambda ^{\ast }q))%
\end{array}%
\right) ,
\end{equation}%
when all the inhomogeneities $\{\eta _{1},...,\eta _{\mathsf{N}}\}$ are pure
phases. Then the transfer matrix $\mathsf{\bar{T}}(\lambda )$ is normal for
any $\lambda \in \mathbb{C}$ and the statement in \textsf{II)} simply
follows.
\end{proof}

\subsection{Antiperiodic spin-1/2 $XXZ$ quantum chain}

In the framework of the quantum inverse scattering method, the integrability
of a quantum models is proven showing that the Hamiltonian of the models
belong to the one parameter families of commuting transfer matrices. In the
following, we solve the spectral problem and compute matrix elements of
local operators on the eigenstates of the antiperiodic transfer matrix $%
\mathsf{\bar{T}}(\lambda )$. It is then relevant to point out that such
analysis allows in particular to describe the antiperiodic spin-1/2 $XXZ$
quantum chain in the special case of the homogeneous limit $\eta
_{n}\rightarrow 1$. Indeed, its Hamiltonian:%
\begin{equation}
H=\sum_{n=1}^{\mathsf{N}}\left[ \sigma _{n}^{x}\sigma _{n+1}^{x}+\sigma
_{n}^{y}\sigma _{n+1}^{y}+\cosh \eta \,\sigma _{n}^{z}\sigma _{n+1}^{z}%
\right] \,,  \label{hamil}
\end{equation}%
with the following boundary conditions:%
\begin{equation}
\sigma _{\mathsf{N}+1}^{a}=\sigma _{1}^{x}\sigma _{1}^{a}\sigma
_{1}^{x}=\left( -1\right) ^{1-\delta _{a,x}}\sigma _{1}^{a},\text{ \ \ }%
a=x,y,z  \label{Antip-boundary}
\end{equation}%
where we have used the notation $\delta _{a,x}=\{1$ for $a=x,$ $0$ for $%
a=y,z\}$, is obtained in the homogeneous limit\footnote{%
Note that substituting in (\ref{logderiv}) the generic transfer matrix $\mathsf{T}^{(\alpha ,b)}$ we get the same Hamiltonian (\ref{hamil}) where the
boundary conditions are given by (\ref{Antip-boundary}) with $\sigma
_{1}^{x} $ substituted by $\Sigma _{1}^{(\alpha ,b)}$.} by:%
\begin{equation}
H=(q-q^{-1})\,\left. \frac{\partial \ln \mathsf{\bar{T}}(\lambda )}{\partial
\lambda }\right\vert _{\lambda =1,\eta _{n}=1}-\mathsf{N}\frac{(q+q^{-1})}{2}%
.  \label{logderiv}
\end{equation}

\section{SOV-representations for $\mathsf{\bar{T}}(\protect\lambda )$%
-spectral problem}

\label{SOV-Gen}According to Sklyanin's method \cite{Sk1}-\cite{Sk3}, a
separation of variable representation for the spectral problem of the
transfer matrix $\mathsf{\bar{T}}(\lambda )$ is defined as a representation
where the commutative family of operators $\mathsf{D}(\lambda )$ (or $%
\mathsf{A}(\lambda )$) is diagonal and with simple spectrum. In fact, it
holds:

\begin{theorem}
For any fixed $\mathsf{N}$-tuple of inhomogeneities $\{\eta _{1},...,\eta _{%
\mathsf{N}}\}\in \mathbb{C}$ $^{\mathsf{N}}$such that%
\begin{equation}
\eta _{a}\neq q^j\eta _{b}\text{ \ }\,\,\forall j\in \{-1,0,1\},\,\, a< b\in \{1,...,\mathsf{N}\}
\label{E-SOV}
\end{equation}%
there exists a SOV representation for the $\mathsf{\bar{T}}(\lambda )$%
-spectral problem; i.e. $\mathsf{D}(\lambda )$\ (or $\mathsf{A}(\lambda )$)\
is diagonalizable and with simple spectrum.
\end{theorem}

The theorem follows by the following explicit construction\footnote{%
For completeness the construction of the SOV representation w.r.t. $\mathsf{A%
}(\lambda )$ is given in appendix.} of $\mathsf{D}(\lambda )$-eigenbasis.
Let us define the left and right \textit{references states}:%
\begin{equation}
\langle 0|\equiv \otimes _{n=1}^{\mathsf{N}}\langle 1,n|\text{ \ \ \ and \ \ 
}|0\rangle \equiv \otimes _{n=1}^{\mathsf{N}}|1,n\rangle ,
\end{equation}%
then:

\begin{theorem}\label{Th-D-SOV}
$\mathsf{I)}$ \underline{Left $\mathsf{D}(\lambda )$ SOV-representations:} \
Under the condition $\left( \ref{E-SOV}\right) $, the states:%
\begin{equation}
\langle h_{1},...,h_{\mathsf{N}}|\equiv \frac{1}{\text{\textsc{n}}}\langle
0|\prod_{n=1}^{\mathsf{N}}\left( \frac{\mathsf{C}(\eta _{n})}{d(\eta _{n}/q)}%
\right) ^{h_{n}},  \label{D-left-eigenstates}
\end{equation}%
where%
\begin{equation}
\text{\textsc{n}}=\prod_{1\leq b<a\leq \mathsf{N}}(\eta _{a}/\eta _{b}-\eta
_{b}/\eta _{a})^{1/2}  \label{Norm-def}
\end{equation}%
\ $h_{n}\in \{0,1\},$ $n\in \{1,...,\mathsf{N}\}$, define a $\mathsf{D}%
(\lambda )$-eigenbasis of $\mathcal{L}_{\mathsf{N}}$:%
\begin{equation}
\langle h_{1},...,h_{\mathsf{N}}|\mathsf{D}(\lambda )=d_{\text{\textbf{h}}%
}(\lambda )\langle h_{1},...,h_{\mathsf{N}}|,  \label{D-L-EigenV}
\end{equation}%
where:%
\begin{equation}
d_{\text{\textbf{h}}}(\lambda )\equiv \prod_{n=1}^{\mathsf{N}}(\frac{\lambda
q^{h_{n}}}{\eta _{n}}-\frac{\eta _{n}}{\lambda q^{h_{n}}})\text{ \ \ \ and \
\ \textbf{h}}\equiv (h_{1},...,h_{\mathsf{N}}).  \label{EigenValue-D}
\end{equation}%
The action of the remaining Yang-Baxter generators on the generic state $%
\langle h_{1},...,h_{\mathsf{N}}|$ reads:%
\begin{eqnarray}
\langle h_{1},...,h_{\mathsf{N}}|\mathsf{C}(\lambda ) &=&\sum_{a=1}^{\mathsf{N}}\prod_{b\neq a}\frac{\lambda
q^{h_{b}}/\eta _{b}-\eta _{b}/q^{h_{b}}\lambda }{\eta
_{a}q^{(h_{b}-h_{a})}/\eta _{b}-\eta _{b}/q^{(h_{b}-h_{a})}\eta _{a}}d(\eta
_{a}q^{h_{a}-1})\langle h_{1},...,h_{\mathsf{N}}|\text{T}_{a}^{+},  \label{C-SOV_D-left} \\
&&  \notag \\
\langle h_{1},...,h_{\mathsf{N}}|\mathsf{B}(\lambda ) &=&\sum_{a=1}^{\mathsf{N}}\prod_{b\neq a}\frac{\lambda
q^{h_{b}}/\eta _{b}-\eta _{b}/q^{h_{b}}\lambda }{\eta
_{a}q^{(h_{b}-h_{a})}/\eta _{b}-\eta _{b}/q^{(h_{b}-h_{a})}\eta _{a}}a(\eta
_{a}q^{h_{a}-1})\langle h_{1},...,h_{\mathsf{N}}|\text{T}_{a}^{-},  \label{B-SOV_D-left}
\end{eqnarray}%
where:%
\begin{equation}
\langle h_{1},...,h_{a},...,h_{\mathsf{N}}|\text{T}_{a}^{\pm }=\langle
h_{1},...,h_{a}\pm 1,...,h_{\mathsf{N}}|.
\end{equation}%
Finally, $\mathsf{A}(\lambda )$ is uniquely defined by the quantum
determinant relation.\smallskip

$\mathsf{II)}$  \underline{Right $\mathsf{D}(\lambda )$ SOV-representations:} \
Under the condition $\left( \ref{E-SOV}\right) $, the states:%
\begin{equation}
|h_{1},...,h_{\mathsf{N}}\rangle \equiv \frac{1}{\text{\textsc{n}}}%
\prod_{n=1}^{\mathsf{N}}\left( \frac{\mathsf{B}(\eta _{n})}{a(\eta _{n})}%
\right) ^{h_{n}}|0\rangle ,  \label{D-right-eigenstates}
\end{equation}%
where\ $h_{n}\in \{0,1\},$ $n\in \{1,...,\mathsf{N}\}$, define a $\mathsf{D}%
(\lambda )$-eigenbasis of $\mathcal{R}_{\mathsf{N}}$:%
\begin{equation}
\mathsf{D}(\lambda )|h_{1},...,h_{\mathsf{N}}\rangle =d_{\text{\textbf{h}}%
}(\lambda )|h_{1},...,h_{\mathsf{N}}\rangle .  \label{D-R-EigenV}
\end{equation}%
The action of the remaining Yang-Baxter generators on the generic state $%
|h_{1},...,h_{\mathsf{N}}\rangle $ reads:%
\begin{eqnarray}
\mathsf{C}(\lambda ) |h_{1},...,h_{\mathsf{N}}\rangle&=&\sum_{a=1}^{\mathsf{N}}\text{T}_{a}^{-}|h_{1},...,h_{\mathsf{N}}\rangle\prod_{b\neq a}\frac{\lambda
q^{h_{b}}/\eta _{b}-\eta _{b}/q^{h_{b}}\lambda }{\eta
_{a}q^{(h_{b}-h_{a})}/\eta _{b}-\eta _{b}/q^{(h_{b}-h_{a})}\eta _{a}}d(\eta
_{a}q^{-h_{a}}),  \label{C-SOV_D-right} \\
&&  \notag \\
\mathsf{B}(\lambda ) |h_{1},...,h_{\mathsf{N}}\rangle&=&\sum_{a=1}^{\mathsf{N}}\text{T}_{a}^{+}|h_{1},...,h_{\mathsf{N}}\rangle\prod_{b\neq a}\frac{\lambda
q^{h_{b}}/\eta _{b}-\eta _{b}/q^{h_{b}}\lambda }{\eta
_{a}q^{(h_{b}-h_{a})}/\eta _{b}-\eta _{b}/q^{(h_{b}-h_{a})}\eta _{a}}a(\eta
_{a}q^{-h_{a}}).  \label{B-SOV_D-right}
\end{eqnarray}%
where:%
\begin{equation}
\text{T}_{a}^{\pm }|h_{1},...,h_{a},...,h_{\mathsf{N}}\rangle
=|h_{1},...,h_{a}\pm 1,...,h_{\mathsf{N}}\rangle .
\end{equation}%
Finally, $\mathsf{A}(\lambda )$ is uniquely defined by the quantum
determinant relation.
\end{theorem}

\begin{proof}
The proof of the theorem is based on Yang-Baxter commutation relations and
on the fact that the left and right references states are $\mathsf{D}(\lambda )$%
-eigenstates:%
\begin{equation}
\langle 0|\mathsf{A}(\lambda )=a(\lambda )\langle 0|,\text{ \ \ \ }\langle 0|%
\mathsf{D}(\lambda )=d(\lambda )\langle 0|,\text{ \ \ \ }\langle 0|\mathsf{B}%
(\lambda )=\text{\b{0}},\text{ \ \ \ }\langle 0|\mathsf{C}(\lambda )\neq 
\text{\b{0}},  \label{L_ref-E}
\end{equation}

and%
\begin{equation}
\mathsf{A}(\lambda )|0\rangle =a(\lambda )|0\rangle ,\text{ \ \ \ }\mathsf{D}%
(\lambda )|0\rangle =d(\lambda )|0\rangle ,\text{ \ \ \ }\mathsf{C}(\lambda
)|0\rangle =\text{\b{0}},\text{ \ \ \ }\mathsf{B}(\lambda )|0\rangle \neq 
\text{\b{0}}.
\end{equation}%
Indeed, to prove that $\left( \ref{D-left-eigenstates}\right) $ and $\left( %
\ref{D-right-eigenstates}\right) $ are left and right eigenstates of $%
\mathsf{D}(\lambda )$ with the eigenvalues $\left( \ref{EigenValue-D}\right) 
$, we have just to repeat the standard computations in algebraic Bethe
ansatz \cite{F95}. Let us see explicitly the left case:%
\begin{align}
\langle h_{1},...,h_{\mathsf{N}}|\mathsf{D}(\lambda )& =\frac{d(\lambda )}{%
\text{\textsc{n}}}\prod_{n=1}^{\mathsf{N}}\left( \frac{\lambda q/\eta
_{n}-\eta _{n}/q\lambda }{\lambda /\eta _{n}-\eta _{n}/\lambda }\right)
^{h_{n}}\langle h_{1},...,h_{\mathsf{N}}|  \notag \\
& +\sum_{n=1}^{\mathsf{N}}\left[ \frac{d(\eta _{n})}{\text{\textsc{n}}}\frac{%
\delta _{h_{n},1}(q-1/q)}{\lambda /\eta _{a}-\eta _{a}/\lambda }\prod_{a\neq
n}\left( \frac{\lambda q/\eta _{a}-\eta _{a}/q\lambda }{\lambda /\eta
_{a}-\eta _{a}/\lambda }\right) ^{h_{n}}\right.  \notag \\
& \times \left. \langle 0|\prod_{a\neq n}\left( \frac{\mathsf{C}(\eta _{a})}{%
d(\eta _{a}/q)}\right) ^{h_{a}}\frac{\mathsf{C}(\lambda )}{d(\eta _{n}/q)}%
\right] ,  \label{Action1}
\end{align}%
where we have used the Yang-Baxter commutation relation:%
\begin{equation}
\mathsf{C}(\mu )\mathsf{D}(\lambda )=\frac{\lambda q/\mu -\mu /q\lambda }{%
\lambda /\mu -\mu /\lambda }\mathsf{D}(\lambda )\mathsf{C}(\mu )+\frac{q-1/q%
}{\lambda /\mu -\mu /\lambda }\mathsf{D}(\mu )\mathsf{C}(\lambda ).
\label{DC-YBC}
\end{equation}%
In particular, the first term on the r.h.s. of $\left( \ref{Action1}\right) $
is generated by using the first term on the r.h.s. of $\left( \ref{DC-YBC}%
\right) $ to commute the operator $\mathsf{D}(\lambda )$ with all the $%
\mathsf{C}(\eta _{a})$ in $\left( \ref{D-left-eigenstates}\right) $
and finally using $\left( \ref{L_ref-E}\right) $ to act with $\mathsf{D}%
(\lambda )$ on the left reference state. The generic term $n$ in the sum of $%
\left( \ref{Action1}\right) $ is obtained by using the commutativity of the $%
\mathsf{C}(\mu )$ to write:%
\begin{equation}
\langle h_{1},...,h_{\mathsf{N}}|=\frac{1}{\text{\textsc{n}}}\langle
0|\left( \frac{\mathsf{C}(\eta _{n})}{d(\eta _{n}/q)}\right)
^{h_{n}}\prod_{a\neq n}\left( \frac{\mathsf{C}(\eta _{a})}{d(\eta _{a}/q)}%
\right) ^{h_{a}},
\end{equation}%
and then commuting $\mathsf{D}(\lambda )$ with all the $\mathsf{C}(\eta
_{a\neq n})$ by using the first term in the r.h.s. of $\left( \ref{DC-YBC}%
\right) $ while for $\mathsf{D}(\lambda )$ and $\mathsf{C}(\eta _{n})$ it is
used the second term in $\left( \ref{DC-YBC}\right) $. The result $\left( %
\ref{D-L-EigenV}\right) $ follows being $d(\eta _{n})=0$.

Under the condition $\left( \ref{E-SOV}\right) $ the states $\langle
h_{1},...,h_{\mathsf{N}}|$ form a set of 2$^{\mathsf{N}}$ independent states
and so they are a $\mathsf{D}(\lambda )$-eigenbasis of the representation.

The action of $\mathsf{B}(\eta _{n}/q^{h_{n}})$ and $\mathsf{C}(\eta
_{n}/q^{h_{n}})$ on the left and right states $\left( \ref%
{D-left-eigenstates}\right) $ and $\left( \ref{D-right-eigenstates}\right) $
follows by imposing the Yang-Baxter commutation relations and the quantum
determinant relations. Then the left $\left( \ref{C-SOV_D-left}\right) $-$%
\left( \ref{B-SOV_D-left}\right) $ and right $\left( \ref{C-SOV_D-right}%
\right) $-$\left( \ref{B-SOV_D-right}\right) $ representations of $\mathsf{B}%
(\lambda )$ and $\mathsf{C}(\lambda )$ are just interpolation formulae which
take into account that they are Laurent polynomials of degree $\mathsf{N}-1 $, respectively even or odd for $\mathsf{N}$ odd or even.
\end{proof}

\textbf{Remark 1.} It is worth remarking that representations of the type $%
\left( \ref{C-SOV_D-left}\right) $-$\left( \ref{B-SOV_D-left}\right) $ and
$\left( \ref{C-SOV_D-right}\right) $-$\left( \ref{B-SOV_D-right}%
\right) $ for the generators of the 6-vertex Yang-Baxter algebra can be also
derived from the original representations by implementing the change of
basis prescribed from the factorizing $F$-matrices \cite{MaiS00}. Let us
recall that these $F$-matrices were introduced to provide explicit
representations of the Drinfel'd's twist of quasi-triangular quasi-Hopf
algebras \cite{Dr1}-\cite{Dr3}. The connection with Sklyanin's quantum
separation of variables in its functional version was recognized in \cite%
{T99} and there used to construct the factorizing $F$-matrices for general
Yangian $Y(sl(2))$; i.e. the rational 6-vertex Yang-Baxter algebra
associated to the general spin s quantum chain representations. These
results were used by the Lyon group in \cite{KitMT99} mainly as tools to get
the solution of the quantum inverse problem and to re-derive the Slavnov's scalar
product formula \cite{Slav89,Slav97} for the periodic spin-1/2 $XXZ$ quantum chain but not to solve
the corresponding spectral problem. This is natural as these representations
do not define quantum separate variables representations for the spectral
problem associated to the transfer matrix of the periodic chain.

\subsection{Sklyanin's measure and scalar products}

In the next subsection, we compute the coupling between states
belonging to right and left SOV-basis. We show that up to an overall
constant these are completely fixed by the left and right
SOV-representations of the Yang-Baxter algebras when the gauge in the
SOV-representations are chosen. Then, we use these results to compute the
scalar products between states which in the left and right SOV-basis have a 
\textit{separated form} similar to that of the transfer matrix eigenstates.
The resulting scalar product formula admits a determinant representation
which can be considered as the SOV analogous of the Slavnov's scalar product
formula computed for Bethe states in the
framework of the algebraic Bethe ansatz.

\subsubsection{Coupling of left and right SOV-basis}
It may be helpful to present the main properties of the matrices which
define the change of basis from the original basis to the SOV-basis; this
will lead us to introduce naturally the coupling between pairs of states
belonging to left and right SOV-basis and the concept of {\it Sklyanin's measure} in our model. Let us define the following isomorphism: 
\begin{equation}
\varkappa :\left( h_{1},...,h_{\mathsf{N}}\right) \in \{0,1\}^{\mathsf{N}%
}\rightarrow j=\varkappa \left( h_{1},...,h_{\mathsf{N}}\right) \equiv
1+\sum_{a=1}^{\mathsf{N}}2^{(a-1)}h_{a}\in \{1,...,2^{\mathsf{N}}\},
\label{corrisp}
\end{equation}%
then we can write:%
\begin{equation}
\langle \text{\textbf{y}}_{j}|=\langle \text{\textbf{x}}_{j}|U^{(L)}=%
\sum_{i=1}^{2^{\mathsf{N}}}U_{j,i}^{(L)}\langle \text{\textbf{x}}_{i}|\text{
\ \ and\ \ \ }|\text{\textbf{y}}_{j}\rangle =U^{(R)}|\text{\textbf{x}}%
_{j}\rangle =\sum_{i=1}^{2^{\mathsf{N}}}U_{i,j}^{(R)}|\text{\textbf{x}}%
_{i}\rangle ,
\end{equation}%
where we have used the notations:%
\begin{equation}
\langle \text{\textbf{y}}_{j}|\equiv \langle h_{1},...,h_{\mathsf{N}}|\text{
\ and \ }|\text{\textbf{y}}_{j}\rangle \equiv |h_{1},...,h_{\mathsf{N}%
}\rangle ,
\end{equation}%
to represent, respectively, the states of the left and right SOV-basis and:%
\begin{equation}
\langle \text{\textbf{x}}_{j}|\equiv \otimes _{n=1}^{\mathsf{N}}\langle
2h_{n}-1,n|\text{ \ \ \ \ and \ \ \ }|\text{\textbf{x}}_{j}\rangle \equiv
\otimes _{n=1}^{\mathsf{N}}|2h_{n}-1,n\rangle ,
\end{equation}%
to represent, respectively, the states of the left and right original $%
\sigma _{n}^{z}$-orthonormal basis. Here, $U^{(L)}$ and $U^{(R)}$ are the $%
2^{\mathsf{N}}\times 2^{\mathsf{N}}$ matrices for which it holds:%
\begin{equation}
U^{(L)}\mathsf{D}(\lambda )=\Delta _{\mathsf{D}}(\lambda )U^{(L)},\text{ \ \ 
}\mathsf{D}(\lambda )U^{(R)}=U^{(R)}\Delta _{\mathsf{D}}(\lambda ),
\end{equation}%
where $\Delta _{\mathsf{D}}(\lambda )$ is a diagonal $2^{\mathsf{N}}\times
2^{\mathsf{N}}$ matrix. The diagonalizability and simplicity of the $\mathsf{%
D}$-spectrum imply the invertibility of the matrices $U^{(L)}$ and $U^{(R)}$%
\ and the fact that all the diagonal entry of $\Delta _{\mathsf{D}}(\lambda
) $ are Laurent polynomials\ in $\lambda $ with different zeros. Then the
following proposition holds:

\begin{proposition}
The $2^{\mathsf{N}}\times 2^{\mathsf{N}}$ matrix:%
\begin{equation}
M\equiv U^{(L)}U^{(R)}
\end{equation}%
is diagonal and it is characterized by:%
\begin{equation}
M_{jj}=\langle \text{\textbf{y}}_{j}|\text{\textbf{y}}_{j}\rangle =\langle
h_{1},...,h_{\mathsf{N}}|h_{1},...,h_{\mathsf{N}}\rangle =\prod_{1\leq
b<a\leq \mathsf{N}}\frac{1}{\eta _{a}q^{(h_{b}-h_{a})}/\eta _{b}-\eta
_{b}/q^{(h_{b}-h_{a})}\eta _{a}}.  \label{M_jj}
\end{equation}
\end{proposition}

\begin{proof}
The fact that the matrix $M$ is diagonal is a trivial consequence of the
{\it orthogonality} of left and right eigenstates corresponding to different
eigenvalue of $\mathsf{D}(\lambda )$.

Let us compute the matrix element $\theta _{a}\equiv \langle
h_{1},...,h_{a}=0,...,h_{\mathsf{N}}|\mathsf{C}(\eta
_{a})|h_{1},...,h_{a}=1,...,h_{\mathsf{N}}\rangle $, where $a\in \{1,...,%
\mathsf{N}\}$. Then using the left action of the operator $\mathsf{C}(\eta
_{a})$ we get:%
\begin{equation}
\theta _{a}=d(\eta _{a}/q)\langle h_{1},...,h_{a}=1,...,h_{\mathsf{N}%
}|h_{1},...,h_{a}=1,...,h_{\mathsf{N}}\rangle,
\end{equation}%
while using the right action of the operator $\mathsf{C}(\eta _{a})$ and the
orthogonality of right and left $\mathsf{D}$-eigenstates corresponding to different
eigenvalues we get:%
\begin{equation}
\theta _{a}=\prod_{b\neq a,b=1}^{\mathsf{N}}\frac{(\eta _{a}q^{h_{b}}/\eta
_{b}-\eta _{b}/q^{h_{b}}\eta _{a})}{(\eta _{a}q^{h_{b}-1}/\eta _{b}-\eta
_{b}/q^{h_{b}-1}\eta _{a})}d(\eta _{a}/q)\langle h_{1},...,h_{a}=0,...,h_{%
\mathsf{N}}|h_{1},...,h_{a}=0,...,h_{\mathsf{N}}\rangle 
\end{equation}%
and so:%
\begin{equation}
\frac{\langle h_{1},...,h_{a}=1,...,h_{\mathsf{N}}|h_{1},...,h_{a}=1,...,h_{%
\mathsf{N}}\rangle }{\langle h_{1},...,h_{a}=0,...,h_{\mathsf{N}%
}|h_{1},...,h_{a}=0,...,h_{\mathsf{N}}\rangle }=\prod_{b\neq a,b=1}^{\mathsf{%
N}}\frac{(\eta _{a}q^{h_{b}}/\eta _{b}-\eta _{b}/q^{h_{b}}\eta _{a})}{(\eta
_{a}q^{h_{b}-1}/\eta _{b}-\eta _{b}/q^{h_{b}-1}\eta _{a})}\,.  \label{F1}
\end{equation}%
The previous formula implies:%
\begin{equation}
\frac{\langle h_{1},...,h_{\mathsf{N}}|h_{1},...,h_{\mathsf{N}}\rangle }{%
\langle 0|0\rangle /\text{\textsc{n}}^{2}}=\prod_{1\leq b<a\leq \mathsf{N}}%
\frac{\eta _{a}/\eta _{b}-\eta _{b}/\eta _{a}}{\eta
_{a}q^{(h_{b}-h_{a})}/\eta _{b}-\eta _{b}/q^{(h_{b}-h_{a})}\eta _{a}},
\label{F2}
\end{equation}%
from which the proposition follows recalling the definition $\left( \ref%
{Norm-def}\right) $ of the normalization \textsc{n}\ and remarking that in
our choice of the original spin basis it holds:%
\begin{equation}
\langle 0|0\rangle =1.
\end{equation}
\end{proof}

\subsubsection{SOV-decomposition of the identity}

The previous results allow to write the following spectral decomposition of
the identity $\mathbb{I}\in\text{End}(\mathcal{R}_{\mathsf{N}})$:%
\begin{equation}
\mathbb{I}\equiv \sum_{i=1}^{2^{\mathsf{N}}}\mu _{i}|\text{\textbf{y}}%
_{i}\rangle \langle \text{\textbf{y}}_{i}|,
\end{equation}%
in terms of the left and right SOV-basis. Here,%
\begin{equation}
\mu _{i}\equiv \frac{1}{\langle \text{\textbf{y}}_{i}|\text{\textbf{y}}%
_{i}\rangle },
\end{equation}%
is the so-called {\it Sklyanin's measure}; which is a discrete measure in the $XXZ$
spin 1/2 quantum chain. Explicitly, the SOV-decomposition of the identity reads: 
\begin{equation}
\mathbb{I}\equiv \sum_{h_{1},...,h_{\mathsf{N}}=0}^{1}\prod_{1\leq b<a\leq 
\mathsf{N}}(\eta_{a}^{2}q^{-2h_{a}}-\eta_{b}^{2}q^{-2h_{b}})\frac{%
|h_{1},...,h_{\mathsf{N}}\rangle \langle h_{1},...,h_{\mathsf{N}}|}{%
\prod_{b=1}^{\mathsf{N}}\omega(\eta_{b}q^{h_{b}})},
\label{Decomp-Id}
\end{equation}%
where:%
\begin{equation}
\omega(\eta)\equiv \eta^{\mathsf{N}-1},
\end{equation}%
and they are gauge dependent parameters.

\textbf{Remark 2.} Sklyanin's measure\footnote{%
See also \cite{Sm98} for further discussions on the measure.} has been first
introduced by Sklyanin in his article on quantum Toda chain \cite{Sk1}.
There it has been derived as a consequence of the self-adjointness of the
transfer matrix w.r.t. the scalar product. In particular, the Hermitian
properties of the operator zeros and their conjugate shift operators have
been fixed to assure the self-adjointness of the transfer matrix. \ In the
similar but more involved non-compact case of the Sinh-Gordon model \cite%
{BT06}, the problem related to the uniqueness of the definition of this
measure has been analyzed. There it has been proven that the measure is in
fact uniquely determined once the positive self-adjointness of the
generators $\mathsf{A}(\lambda )$ and $\mathsf{D}(\lambda )$ is required. The approach here used is suitable for general compact SOV-representations of 6-vertex Yang-Baxter algebra.

\subsection{SOV characterization of $\mathsf{\bar{T}}(\protect\lambda )$%
-spectrum}

Let us denote with $\Sigma _{\mathsf{\bar{T}}}$ the set of the eigenvalue
functions $t(\lambda )$ of the transfer matrix $\mathsf{\bar{T}}(\lambda )$,
then $\Sigma _{\mathsf{\bar{T}}}$ is contained in: 
\begin{equation}
\mathbb{C}_{even}[\lambda ,\lambda ^{-1}]_{\mathsf{N}-1}\text{ for }\mathsf{N%
}\text{ odd, \ \ \ }\mathbb{C}_{odd}[\lambda ,\lambda ^{-1}]_{\mathsf{N}-1}%
\text{ for }\mathsf{N}\text{ even},  \label{set-t}
\end{equation}%
where $\mathbb{C}_{\epsilon }[x,x^{-1}]_{\mathsf{M}}$ denotes the linear
space in the field $\mathbb{C}$ of the Laurent polynomials of degree $%
\mathsf{M}$ in the variable $x$ which are even or odd as stated in the index 
$\epsilon $.

\begin{theorem}
\label{C:T-eigenstates}If the inhomogeneities parameters $\{\eta _{1},\dots
,\eta _{\mathsf{N}}\}$ satisfy the conditions $\left( \ref{E-SOV}\right) $,
then the spectrum of $\mathsf{\bar{T}}(\lambda )$ is simple and $\Sigma _{{%
\mathsf{\bar{T}}}}$ coincides with the set of functions in (\ref{set-t})
which are solutions of the discrete system of equations:%
\begin{equation}
t(\eta _{a})t(\eta _{a}/q)=a(\eta _{a})d(\eta _{a}/q),\text{ \ \ }\forall
a\in \{1,...,\mathsf{N}\}.  \label{I-Functional-eq}
\end{equation}

\begin{itemize}
\item[\textsf{I)}] The right $\mathsf{\bar{T}}$-eigenstate corresponding to
a $t(\lambda )\in \Sigma _{\mathsf{\bar{T}}}$ is characterized by:%
\begin{equation}
|t\rangle =\sum_{h_{1},...,h_{\mathsf{N}}=0}^{1}\prod_{a=1}^{\mathsf{N}}%
\frac{Q_{t}(\eta _{a}q^{-h_{a}})}{\omega(\eta _{a}q^{-h_{a}})}%
\prod_{1\leq b<a\leq \mathsf{N}}(\eta^{2}_{a}q^{-2h_{a}}-\eta^{2}_{b}q^{-2h_{b}})|h_{1},...,h_{\mathsf{N}}\rangle
,  \label{eigenT-r-D}
\end{equation}%
where, up to an overall normalization of the state, the coefficients are
characterized by:%
\begin{equation}
Q_{t}(\eta _{a}/q)/Q_{t}(\eta _{a})=t(\eta _{a})/d(\eta _{a}/q).
\label{t-Q-relation}
\end{equation}

\item[\textsf{II)}] The left $\mathsf{\bar{T}}$-eigenstate corresponding to $%
t(\lambda )\in \Sigma _{\mathsf{\bar{T}}}$ is characterized by:%
\begin{equation}
\langle t|=\sum_{h_{1},...,h_{\mathsf{N}}=0}^{1}\prod_{a=1}^{\mathsf{N}}%
\frac{\bar Q_{t}(\eta _{a}q^{-h_{a}})}{\omega(\eta _{a}q^{-h_{a}})}%
\prod_{1\leq b<a\leq \mathsf{N}}(\eta^{2}_{a}q^{-2h_{a}}-\eta^{2}_{b}q^{-2h_{b}})\langle h_{1},...,h_{\mathsf{N}%
}|,  \label{eigenT-l-D}
\end{equation}%
where up to an overall normalization of the state the coefficients are
characterized by:%
\begin{equation}
\bar{Q}_{t}(\eta _{a}/q)/\bar{Q}_{t}(\eta _{a})=t(\eta _{a})/a(\eta _{a}).
\label{t-Qbar-relation}
\end{equation}
\end{itemize}
\end{theorem}

\begin{proof}
In the SOV representations the spectral problem for {$\mathsf{\bar{T}}$}$%
(\lambda )$ is reduced to a discrete system of $2^{\mathsf{N}}$ Baxter-like
equations. Indeed, let $\left\langle t\right\vert $ be a $\mathsf{\bar{T}}$%
-eigenstate corresponding to the eigenvalue $t(\lambda )\in \Sigma _{{%
\mathsf{\bar{T}}}}$, then the coefficients (\textit{wave-functions})%
\begin{equation}
\Psi _{t}(\text{\textbf{h}})\equiv \left\langle t\right\vert h_{1},...,h_{%
\mathsf{N}}\rangle
\end{equation}%
of $\left\langle t\right\vert $ in the SOV-basis satisfy the equations: 
\begin{equation}
t(\eta _{n}q^{-h_{n}})\Psi _{t}(\text{\textbf{h}})\,=\,a(\eta
_{n}q^{-h_{n}})\Psi _{t}(\mathsf{T}_{n}^{+}(\text{\textbf{h}}))+d(\eta
_{n}q^{-h_{n}})\Psi _{t}(\mathsf{T}_{n}^{-}(\text{\textbf{h}})),
\label{SOVBax1}
\end{equation}%
for any$\,n\in \{1,...,\mathsf{N}\}$ and \textbf{h}$\in \{0,1\}^{\mathsf{%
N}}$, where we have denoted:%
\begin{equation}
\mathsf{T}_{n}^{\pm }(\text{\textbf{h}})\equiv (h_{1},\dots ,h_{n}\pm
1,\dots ,h_{\mathsf{N}}).
\end{equation}%
Taking into account that:%
\begin{equation}
a(\eta _{n}/q)=d(\eta _{n})=0,
\end{equation}%
the previous system of equations can be rewritten as a system of homogeneous
equations:%
\begin{equation}
\left( 
\begin{array}{cc}
t(\eta _{n}) & -a(\eta _{n}) \\ 
-d(\eta _{n}/q) & t(\eta _{n}/q)%
\end{array}%
\right) \left( 
\begin{array}{c}
\Psi _{t}(h_{1},...,h_{n}=0,...,h_{1}) \\ 
\Psi _{t}(h_{1},...,h_{n}=1,...,h_{1})%
\end{array}%
\right) =\left( 
\begin{array}{c}
0 \\ 
0%
\end{array}%
\right) ,  \label{homo-system}
\end{equation}%
for any$\,n\in \{1,...,\mathsf{N}\}$ with $h_{m\neq n}\in \{0,1\}$. Note
that the condition $t(\lambda )\in \Sigma _{{\mathsf{\bar{T}}}}$ implies
that the previous system has to have a non-trivial solution, i.e. the
determinants of the $2\times 2$ matrices in $\left( \ref{homo-system}\right) 
$ must be zero for any$\,n\in \{1,...,\mathsf{N}\}$. So that the condition $%
t(\lambda )\in \Sigma _{{\mathsf{\bar{T}}}}$ implies (\ref{I-Functional-eq}%
). Now let us observe that being%
\begin{equation}
a(\eta _{n})\neq 0\text{\ \ and \ }d(\eta _{n}/q)\neq 0,  \label{Rank1}
\end{equation}%
the rank of the matrices in $\left( \ref{homo-system}\right) $ is 1 and then
up to an overall normalization the solution is unique:%
\begin{equation}
\frac{\Psi _{t}(h_{1},...,h_{n}=1,...,h_{1})}{\Psi
_{t}(h_{1},...,h_{n}=0,...,h_{1})}=\frac{t(\eta _{a})}{a(\eta _{a})},
\end{equation}%
for any$\,n\in \{1,...,\mathsf{N}\}$ with $h_{m\neq n}\in \{0,1\}$. This
implies that given a $t(\lambda )\in \Sigma _{{\mathsf{\bar{T}}}}$ there
exist (up to normalization) one and only one corresponding $\mathsf{\bar{T}}$%
-eigenstate $\left\langle t\right\vert $ with coefficients which have the
factorized form given in $\left( \ref{eigenT-l-D}\right) $-$\left( \ref%
{t-Qbar-relation}\right) $ and then the $\mathsf{\bar{T}}$-spectrum is
simple.

Vice versa, let us take a $t(\lambda )$ in the set of functions (\ref{set-t}%
) which is solution of the system (\ref{I-Functional-eq}) then for the state 
$\left\langle t\right\vert $ constructed by $\left( \ref{eigenT-l-D}\right) $%
-$\left( \ref{t-Qbar-relation}\right) $ it holds:%
\begin{equation}
\left\langle t\right\vert \mathsf{\bar{T}}(\eta _{n}q^{-h_{n}})|h_{1},...,h_{%
\mathsf{N}}\rangle =t(\eta _{n}q^{-h_{n}})\left\langle t\right\vert
h_{1},...,h_{\mathsf{N}}\rangle \text{ \ }\forall n\in \{1,...,\mathsf{N}\},
\end{equation}
for any $\mathsf{D}$-eigenstate $|h_{1},...,h_{\mathsf{N}}\rangle $.
Then being $\mathsf{\bar{T}}(\lambda )$ a Laurent polynomials of degree 
$\mathsf{N}-1$ in $\lambda $, respectively even or odd for $\mathsf{N}$
odd or even, it follows:%
\begin{equation}
\left\langle t\right\vert \mathsf{\bar{T}}(\lambda )|h_{1},...,h_{\mathsf{N}%
}\rangle =t(\lambda )\left\langle t\right\vert h_{1},...,h_{\mathsf{N}%
}\rangle ,
\end{equation}%
that is $t(\lambda )\in \Sigma _{{\mathsf{\bar{T}}}}$\ and $\left\langle
t\right\vert $ is the corresponding $\mathsf{\bar{T}}$-eigenstate.
\end{proof}

The previous theorem gives a well defined and complete
characterization of the spectrum of the transfer matrix $\mathsf{\bar{T}}%
(\lambda )$ and the normality of $\mathsf{\bar{T}}(\lambda )$ implies that
the discrete system of equations has to admit 2$^{\mathsf{N}}$ independent
solutions in the class of functions (\ref{set-t}). However, it is worth
pointing out that such a characterization of the spectrum is not the most
efficient; in particular, for the analysis of the continuum limit. A
reformulation of the SOV characterization of the $\mathsf{\bar{T}}$-spectrum
by functional equations is then important and it can be achieved by the
construction of a Baxter Q-operator whose functional equation, computed in
the spectrum of the $\mathsf{D}$-zeros, coincides with the finite system of
Baxter-like equations $\left( \ref{SOVBax1}\right) $. For the model under
consideration a Q-operator has been constructed in \cite{BBOY95} and
it satisfies\footnote{%
If we translate it in our notations and we introduce the inhomogeneities.}:%
\begin{equation}
\left[ \mathsf{\bar{T}}(\lambda ),Q(\lambda )\right] =0,\text{ }\left[
Q(\lambda ),Q(\mu )\right] =0,\text{ \ }\mathsf{\bar{T}}(\lambda )Q(\lambda
)=a\left( \lambda \right) Q(\lambda /q)+d\left( \lambda \right) Q(\lambda
q)\,.  \label{Q-op-ch}
\end{equation}%
Moreover, $Q(\lambda )$ is a Laurent polynomial with eigenvalues of the form:%
\begin{equation}
\text{\textsc{q}}(\lambda )=\lambda ^{-\mathsf{N}/2}\prod_{k=1}^{\mathsf{N}%
}(\lambda -\lambda _{k})\,,  \label{Q-eigen}
\end{equation}%
where the $\{\lambda _{j}\}$ are solutions of the Bethe equations:%
\begin{equation}
\prod_{n=1}^{\mathsf{N}}\frac{q^{2}\lambda _{k}^{2}-\eta _{n}^{2}}{\lambda
_{k}^{2}-\eta _{n}^{2}}=-\prod_{a=1}^{\mathsf{N}}\frac{q\lambda _{k}-\lambda
_{a}}{\lambda _{k}/q-\lambda _{a}}\text{ \ \ \ }\forall k\in \{1,...,\mathsf{%
N}\},  \label{Bethe-Ansatz}
\end{equation}%
as a consequence of the requirement of analyticity of the transfer matrix
eigenvalues.

Let us observe that for $q=e^{2i\pi p^{\prime }/p}$ ($p,$ $p^{\prime }\in
Z^{\geq 0}$) a $p$-root of unit the consistence condition (i.e. the
existence of non-trivial solutions) for the Baxter eigenvalue functional
equation lead to the functional equation\footnote{%
This characterization open the possibility to construct solutions of the
Baxter equation in terms of cofactors of the matrix $D(\lambda )$ as
explained in \cite{N-10}.}: 
\begin{equation}
\det_{p}D(\Lambda )=0,\text{ \ \ }\Lambda \in \mathbb{C}
\label{fun-eq-T-eigen}
\end{equation}%
involving only the $\mathsf{\bar{T}}$-eigenvalue $t(\lambda )$, where $%
\Lambda =\lambda ^{p}$ and $D(\lambda )$ is the $p\times p$ matrix:%
\begin{equation}
D(\lambda )\equiv 
\begin{pmatrix}
t(\lambda ) & -d(\lambda ) & 0 & \cdots & 0 & -a(\lambda ) \\ 
-a(q\lambda ) & t(q\lambda ) & -d(q\lambda ) & 0 & \cdots & 0 \\ 
0 & {\quad }\ddots &  &  &  & \vdots \\ 
\vdots &  & \cdots &  &  & \vdots \\ 
\vdots &  &  & \cdots &  & \vdots \\ 
\vdots &  &  &  & \ddots {\qquad } & 0 \\ 
0 & \ldots & 0 & -a(q^{p-2}\lambda ) & t(q^{p-2}\lambda ) & 
-d(q^{p-2}\lambda ) \\ 
-d(q^{p-1}\lambda ) & 0 & \ldots & 0 & -a(q^{p-1}\lambda ) & 
t(q^{p-1}\lambda )%
\end{pmatrix}%
.  \label{D-matrix}
\end{equation}%
It is worth remarking that the equation $\left( \ref{fun-eq-T-eigen}\right) $
coincides with the equation obtained by combining the fusion of transfer
matrices \cite{KRS,KR}\ and the truncation identity which holds at the root
of unit\footnote{%
This method was first developed for the RSOS model in \cite{BR89} and was
adapted to spin chains in \cite{N02,N03}.}.

\section{Scalar products}\label{SP}

Let $\langle \alpha |$ and $|\beta \rangle $ be two arbitrary left and right
{\it separate states} which by definition have the following factorized form in the SOV-representation:%
\begin{align}
\langle \alpha |& =\sum_{h_{1},...,h_{\mathsf{N}}=0}^{1}\prod_{a=1}^{\mathsf{%
N}}\frac{\alpha _{a}(\eta _{a}q^{-h_{a}})}{\omega(\eta _{a}q^{-h_{a}})}%
\prod_{1\leq b<a\leq \mathsf{N}}(\eta^{2}_{a}q^{-2h_{a}}-\eta^{2}_{b}q^{-2h_{b}})\langle h_{1},...,h_{\mathsf{N}%
}|,  \label{Fact-left-SOV} \\
|\beta \rangle & =\sum_{h_{1},...,h_{\mathsf{N}}=0}^{1}\prod_{a=1}^{\mathsf{N%
}}\frac{\beta _{a}(\eta _{a}q^{-h_{a}})}{\omega(\eta _{a}q^{-h_{a}})}%
\prod_{1\leq b<a\leq \mathsf{N}}(\eta^{2}_{a}q^{-2h_{a}}-\eta^{2}_{b}q^{-2h_{b}})|h_{1},...,h_{\mathsf{N}}\rangle
,  \label{Fact-right-SOV}
\end{align}%
The interest toward these kind of states is due to the following:

\begin{proposition}The two states\footnote{%
Note that by using the Hermitian conjugation properties of the Yang-Baxter
generators the vector $\left( \langle \alpha |\right) ^{\dagger }\in \mathcal{%
R}_{\mathsf{N}}$ can simply be written in the $\mathsf{A}(\lambda )$
SOV-basis, see appendix for the definition of this basis. The formula %
\rf{scalar-product-general} also describes the action of the covector $%
\langle \alpha |$ on the vector $|\beta \rangle $.} $\left( \langle \alpha
|\right) ^{\dagger }$ and $|\beta \rangle $ have the following scalar product:
\begin{equation}\label{scalar-product-general}
\langle \alpha |\beta \rangle =\det_{\mathsf{N}}||\mathcal{M}_{a,b}^{\left(
\alpha ,\beta \right) }||\text{ \ \ with \ }\mathcal{M}_{a,b}^{\left( \alpha
,\beta \right) }\equiv \left( \eta _{a}\right) ^{2(b-1)}\sum_{h=0}^{1}\frac{%
\alpha _{a}(\eta _{a}q^{-h})\beta _{a}(\eta _{a}q^{-h})}{\omega(\eta
_{a}q^{-h})}q^{-2(b-1)h}.
\end{equation}
\end{proposition}

\begin{proof}
From the SOV-decomposition, we have:%
\begin{equation}
\langle \alpha |\beta \rangle =\sum_{h_{1},...,h_{\mathsf{N}}=0}^{1}V(\left(
\eta _{1}/q^{h_{1}}\right) ^{2},...,\left( \eta _{\mathsf{N}}/q^{h_{\mathsf{N%
}}}\right) ^{2})\prod_{a=1}^{\mathsf{N}}\frac{\alpha _{a}(\eta
_{a}q^{-h_{a}})\beta _{a}(\eta _{a}q^{-h_{a}})}{\omega(\eta
_{a}q^{-h_{a}})},
\end{equation}%
where $V(x_{1},...,x_{\mathsf{N}})\equiv \prod_{1\leq b<a\leq \mathsf{N}%
}(x_{a}-x_{b})$ is the Vandermonde determinant. From this formula by using
the multilinearity of the determinant w.r.t. the rows we prove the
proposition.
\end{proof}

Note that the form of determinant for these scalar products is not restricted to the case in which one of the two
states is an eigenstate of the transfer matrix; on the contrary to what
happens for the scalar product formulae in the framework of the algebraic
Bethe ansatz. It is worth noticing that from the scalar product formula we
can prove:

\begin{corollary}
Transfer matrix eigenstates corresponding to different eigenvalues are
orthogonal states.
\end{corollary}

\begin{proof}
Let us denote with $|t\rangle $ and $|t^{\prime }\rangle $ two eigenstates
of $\mathsf{\bar{T}}(\lambda )$ with eigenvalues $t(\lambda )$ and $%
t^{\prime }(\lambda )$. To prove the corollary, we have to prove that:%
\begin{equation}
\det_{\mathsf{N}}||\Phi _{a,b}^{\left( t,t^{\prime }\right) }||=0\text{ \ \
with \ }\Phi _{a,b}^{\left( t,t^{\prime }\right) }\equiv \left( \eta
_{a}\right) ^{2(b-1)}\sum_{c=0}^{1}\frac{Q_{t^{\prime }}(\eta _{a}q^{-c})%
\bar{Q}_{t}(\eta _{a}q^{-c})}{\omega(\eta _{a}q^{-c})}q^{-2(b-1)c}.
\label{orth-cond}
\end{equation}%
It is enough to show the existence of a non-zero vector V$^{\left(
t,t^{\prime }\right) }$ such that:%
\begin{equation}
\sum_{b=1}^{\mathsf{N}}\Phi _{a,b}^{\left( t,t^{\prime }\right) }\text{V}%
_{b}^{\left( t,t^{\prime }\right) }=0\text{ \ \ \ \ }\forall a\in \{1,...,%
\mathsf{N}\}.  \label{zero-eigenvector}
\end{equation}%
The transfer matrix eigenvalues are Laurent polynomials of degree $\mathsf{N}
$-1 (even for $\mathsf{N}$-1 even and odd for $\mathsf{N}$-1 odd) of the
form:%
\begin{equation}
t(\lambda )=\sum_{b=1}^{\mathsf{N}}c_{b}\lambda ^{-\mathsf{N}-1+2b},\text{ \
\ }t^{\prime }(\lambda )=\sum_{b=1}^{\mathsf{N}}c_{b}^{\prime }\lambda ^{-%
\mathsf{N}-1+2b},
\end{equation}%
so if we define:%
\begin{equation}
\text{V}_{b}^{\left( t,t^{\prime }\right) }\equiv c_{b}^{\prime }-c_{b}\text{%
\ \ \ }\forall b\in \{1,...,\mathsf{N}\},
\end{equation}%
it results:%
\begin{equation}
\sum_{b=1}^{\mathsf{N}}\Phi _{a,b}^{\left( t,t^{\prime }\right) }\text{V}%
_{b}^{\left( t,t^{\prime }\right) }=\sum_{c=0}^{1}Q_{t^{\prime }}(\eta
_{a}q^{-c})\bar{Q}_{t}(\eta _{a}q^{-c})(t^{\prime }(\eta _{a}q^{-c})-t(\eta
_{a}q^{-c})).  \label{zero-eigenvector-1}
\end{equation}%
We can use now the discrete system of Baxter equations satisfied by the $%
Q_{t^{\prime }}(\eta _{a}q^{-h_{a}})$ and $\bar{Q}_{t}(\eta _{a}q^{-h_{a}})$
to rewrite:%
\begin{equation}
Q_{t^{\prime }}(\eta _{a})\bar{Q}_{t}(\eta _{a})(t^{\prime }(\eta
_{a})-t(\eta _{a}))=a(\eta _{a})Q_{t^{\prime }}(\eta _{a}/q)\bar{Q}_{t}(\eta
_{a})-d(\eta _{a}/q)Q_{t^{\prime }}(\eta _{a})\bar{Q}_{t}(\eta _{a}/q),
\end{equation}%
and%
\begin{equation}
Q_{t^{\prime }}(\eta _{a}/q)\bar{Q}_{t}(\eta _{a}/q)(t^{\prime }(\eta
_{a}/q)-t(\eta _{a}/q))=d(\eta _{a}/q)Q_{t^{\prime }}(\eta _{a})\bar{Q}%
_{t}(\eta _{a}/q)-a(\eta _{a})Q_{t^{\prime }}(\eta _{a}/q)\bar{Q}_{t}(\eta
_{a})
\end{equation}%
and by substituting them in (\ref{zero-eigenvector-1}) we get (\ref%
{zero-eigenvector}).
\end{proof}

\section{Reconstruction of local operators}

The first reconstruction of local operators has been achieved in \cite%
{KitMT99}, for the case of the $XXZ$ spin 1/2 chain. In \cite{MaiT00}, then
the solution has been extended to fundamental lattice models, i.e. those
with isomorphic\ auxiliary and local quantum space, for which the monodromy
matrix becomes the permutation operator at a special value of the spectral
parameter. Here, we present a simple modification of the reconstruction
formula of \cite{KitMT99} to adapt it to the current antiperiodic
case.

\begin{proposition}
The following reconstruction holds in terms of the antiperiodic transfer
matrix:%
\begin{eqnarray}
X_{n} &=&\prod_{b=1}^{n-1}\mathsf{\bar{T}}(\eta _{b})\text{tr}_{0}(\mathsf{M}%
_{0}(\eta _{n})X_{0}\sigma _{0}^{x})\prod_{b=1}^{n}\frac{\mathsf{\bar{T}}%
(\eta _{b}/q)}{\det {}\mathsf{\bar{M}}(\eta _{b})}  \label{R-AP-1} \\
&=&\prod_{b=1}^{n}\mathsf{\bar{T}}(\eta _{b})\frac{\text{tr}_{0}(\sigma
_{0}^{(z)}\mathsf{M}_{0}^{t_{0}}(\eta _{n}/q)\sigma _{0}^{z}X_{0}\sigma
_{0}^{x})}{\det \mathsf{M}(\eta _{n})}\prod_{b=1}^{n-1}\frac{\mathsf{\bar{T}}%
(\eta _{b}/q)}{\det {}\mathsf{\bar{M}}(\eta _{b})},  \label{R-AP-2}
\end{eqnarray}%
where it holds:%
\begin{equation}
\mathsf{\bar{T}}(\eta _{n})\mathsf{\bar{T}}(\eta _{n}/q)=\det {}\mathsf{\bar{%
M}}(\eta _{b})=\mathsf{B}(\eta _{b})\mathsf{C}(\eta _{b}/q)-\mathsf{A}(\eta
_{b})\mathsf{D}(\eta _{b}/q).  \label{R-AP-3}
\end{equation}
\end{proposition}

\begin{proof}
Here, we present a proof based on the following known \cite{KitMT99}
reconstruction in terms of the periodic transfer matrix:%
\begin{eqnarray}
X_{n} &=&\prod_{b=1}^{n-1}\mathsf{T}(\eta _{b})\text{tr}_{0}(\mathsf{M}%
_{0}(\eta _{n})X_{0})\prod_{b=1}^{n}\frac{\mathsf{T}(\eta _{b}/q)}{\det 
\mathsf{M}(\eta _{b})}  \label{R-P-1} \\
&=&\prod_{b=1}^{n}\mathsf{T}(\eta _{b})\frac{\text{tr}_{0}(\sigma _{0}^{y}%
\mathsf{M}_{0}^{t_{0}}(\eta _{n}/q)\sigma _{0}^{y}X_{0})}{\det \mathsf{M}%
(\eta _{n})}\prod_{b=1}^{n-1}\frac{\mathsf{T}(\eta _{b}/q)}{\det \mathsf{M}%
(\eta _{b})},  \label{R-P-2}
\end{eqnarray}%
where it holds:%
\begin{equation}
\mathsf{T}(\eta _{b})\mathsf{T}(\eta _{b}/q)=\det {}\mathsf{M}(\eta _{b})=%
\mathsf{A}(\eta _{b})\mathsf{D}(\eta _{b}/q)-\mathsf{B}(\eta _{b})\mathsf{C}%
(\eta _{b}/q).  \label{R-P-3}
\end{equation}%
From the above formulae it holds:%
\begin{equation}
\sigma _{n}^{x}\overset{(\ref{R-P-1})}{=}\prod_{b=1}^{n-1}\mathsf{T}(\eta
_{b})\mathsf{\bar{T}}(\eta _{n})\prod_{b=1}^{n}\frac{\mathsf{T}(\eta _{b}/q)%
}{\det \mathsf{M}(\eta _{b})}\overset{(\ref{R-P-2})}{=}-\prod_{b=1}^{n}\frac{%
\mathsf{T}(\eta _{b})}{\det \mathsf{M}(\eta _{b})}\mathsf{\bar{T}}(\eta
_{n}/q)\prod_{b=1}^{n-1}\mathsf{T}(\eta _{b}/q),  \label{s^x_n}
\end{equation}%
from which:%
\begin{equation}
1=\sigma _{n}^{x}\sigma _{n}^{x}=-\prod_{b=1}^{n-1}\mathsf{T}(\eta _{b})%
\frac{\mathsf{\bar{T}}(\eta _{n})\mathsf{\bar{T}}(\eta _{n}/q)}{\det \mathsf{%
M}(\eta _{n})}\prod_{b=1}^{n-1}\frac{\mathsf{T}(\eta _{b}/q)}{\det \mathsf{M}%
(\eta _{b})},
\end{equation}%
which implies $(\ref{R-AP-3})$ thanks to $(\ref{R-P-3})$. Moreover, we can
use $(\ref{s^x_n})$ to write:%
\begin{equation}
\prod_{b=1}^{c-1}\sigma _{b}^{x}=\prod_{b=1}^{c-1}\mathsf{\bar{T}}(\eta
_{b})\prod_{b=1}^{c-1}\frac{\mathsf{T}(\eta _{b}/q)}{\det \mathsf{M}(\eta
_{b})}=\prod_{b=1}^{c-1}\frac{\mathsf{T}(\eta _{b})}{\det \mathsf{\bar{M}}%
(\eta _{b})}\prod_{b=1}^{c-1}\mathsf{\bar{T}}(\eta _{n}/q),  \label{P-s^x_n}
\end{equation}%
then the result $(\ref{R-AP-1})$ follows by computing:%
\begin{equation}
X_{n}=\prod_{b=1}^{n-1}\sigma _{b}^{x}\tilde{X}_{n}\prod_{b=1}^{n}\sigma
_{b}^{x}\text{ \ \ with \ \ }\tilde{X}_{n}=X_{n}\sigma _{n}^{x}
\end{equation}%
by using for $\tilde{X}_{n}$ the reconstruction $(\ref{R-P-1})$ and for the
first product of $\sigma _{b}^{x}$ the first reconstruction in $(\ref%
{P-s^x_n})$ while for the second product of $\sigma _{b}^{x}$ the
second reconstruction in $(\ref{P-s^x_n})$. Similarly, the result $(\ref{R-AP-2})$
follows by computing:%
\begin{equation}
X_{n}=\prod_{b=1}^{n}\sigma _{b}^{x}\bar{X}_{n}\prod_{b=1}^{n-1}\sigma
_{b}^{x}\text{ \ \ with \ \ }\bar{X}_{n}=\sigma _{n}^{x}X_{n},
\end{equation}%
by using for $\bar{X}_{n}$ the reconstruction $(\ref{R-P-2})$.
\end{proof}

\section{Form factors of the local operators}
\subsection{Preliminary comments}

In the following we will compute the matrix elements (form factors):%
\begin{equation}
\langle t|O_{n}|t^{\prime }\rangle   \label{General-ME}
\end{equation}
which by definition are the action of the transfer matrix eigencovector $%
\langle t|\in \mathcal{L}_{\mathsf{N}}$ on the vector obtained by the action
of some local spin operator $O_{n}$ on the transfer matrix eigenvector $%
|t^{\prime }\rangle \in \mathcal{R}_{\mathsf{N}}$, where $|t^{\prime }\rangle$ and $\langle t|$ are
defined in \rf{eigenT-r-D}  and  \rf{eigenT-l-D}. As explicitly stated in Theorem \ref{C:T-eigenstates}, these states are by
definition characterized up to an overall normalization, then it is worth
pointing out that these normalizations do not lead to limitations in the use
of the form factors $\left( \ref{General-ME}\right) $ to expand m-point
functions like:
\begin{equation}
\frac{\langle t|O_{n_{1}}\cdots O_{n_{\text{m}}}|t\rangle }{\langle t|t\rangle }
,  \label{General-n-point-F}
\end{equation}%
where we are denoting with $\langle t|t\rangle $ the action of the covector $\langle t|$ on the vector $|t\rangle $ as defined in \rf{scalar-product-general} and as well as with $\langle t|O_{n_{1}}\cdots O_{n_{\text{m}}}|t\rangle$ the action of the covector $%
\langle t|$ on the vector $O_{n_{1}}\cdots O_{n_{\text{m}}}|t\rangle$. Indeed, it is enough
to remark that the m-point functions of the type $\left( \ref%
{General-n-point-F}\right) $ are normalization independent and moreover the
following formula:%
\begin{equation}\label{Id-decomp}
\mathbb{I=}\sum_{t(\lambda )\in \sum_{\mathsf{T}}}\frac{|t\rangle \langle t|%
}{\langle t|t\rangle },
\end{equation}%
is a well defined decomposition of the identity from the diagonalizability
and simplicity of the transfer matrix spectrum. Then, we can write:
\begin{equation}\label{FF-expansion}
\frac{\langle t|O_{n_{1}}\cdots O_{n_{\text{m}}}|t\rangle }{\langle t|t\rangle }%
=\sum_{t_{1}(\lambda ),...,t_{\text{m}-1}(\lambda )\in \sum_{\mathsf{T}}}%
\frac{\langle t|O_{n_{1}}|t_{1}\rangle \langle t_{\text{m}-1}|O_{n_{\text{m}%
}}|t\rangle \prod_{a=2}^{\text{m}-1}\langle t_{a-1}|O_{n_{a}}|t_{a}\rangle }{%
\langle t|t\rangle \prod_{a=1}^{\text{m}-1}\langle t_{a}|t_{a}\rangle },
\end{equation}%
where in the r.h.s there are the matrix elements that we compute in this
paper.

In the representations which defines a normal transfer matrix $\mathsf{\bar{T%
}}\left( \lambda \right) $ it is worth remarking that the issue of the
relative normalization between eigencovector and eigenvector of $\mathsf{%
\bar{T}}\left( \lambda \right) $ becomes more important. Indeed, taken the
generic eigenvector $|t\rangle $\ then the covector \underline{$\langle t|$}$%
\equiv \left( |t\rangle \right) ^{\dagger }$, dual to $|t\rangle $ w.r.t.
the Hermitian conjugation $^{\dagger }$, is itself an eigencovector of $%
\mathsf{\bar{T}}\left( \lambda \right) $ which for the simplicity of $%
\mathsf{\bar{T}}$-spectrum implies the following
identity \underline{$\langle t|$}$\equiv \alpha _{t}\langle t|$, where $%
\langle t|$ is the eigencovector defined in \rf{eigenT-l-D}. Of course, in these representations the following identities hold:
\begin{equation}
\frac{\langle t|O_{n_{1}}\cdots O_{n_{\text{m}}}|t\rangle }{\langle t|t\rangle }%
=\frac{\underline{\langle t|}O_{n_{1}}\cdots O_{n_{\text{m}}}|t\rangle }{%
\left\Vert |t\rangle \right\Vert ^{2}}
\end{equation}
where $\left\Vert |t\rangle
\right\Vert$ is the positive norm of the eigenvector $|t\rangle$ in the Hilbert space $\mathcal{R}_{\mathsf{N}}$ w.r.t. the scalar product introduced in Section \ref{v-scalar-product}. Then the form factor expansion defined in \rf{FF-expansion} can be used as well to compute the m-point functions for the standard definition in the Hilbert space $\mathcal{R}_{\mathsf{N}}$.

Let us finally comment that from the above discussion emerges clearly the relevance to compute explicitly the norm of the transfer matrix eigenvectors $|t\rangle $, defined in \rf{eigenT-r-D}, as it allows to fix the relative normalization $\alpha _{t}$ of left and right transfer matrix eigenstates thanks to the identity $\alpha _{t}=\left\Vert |t\rangle
\right\Vert ^{2}/\langle t|t\rangle $ in this way making possible to take these left and right states as one the exact dual of the other; this interesting issue is currently under analysis.

\subsection{Results}
Here, we present the main results of the present paper:

\begin{theorem}
\label{FF-Prop1}Let $\langle t|$ and $|t^{\prime }\rangle $ be a left and a right eigenstate of the transfer matrix $\mathsf{\bar{T}}(\lambda )$, respectively, then it
holds:%
\begin{equation}
\langle t|\sigma
_{n}^{-}|t^{\prime }\rangle =\frac{\prod_{h=1}^{n-1}t(\eta
_{h})\prod_{h=1}^{n}t^{\prime }(\eta _{h}/q)}{\prod_{h=1}^{n}a(\eta _{h})d(\eta _{h}/q)}\det_{\mathsf{N}+1}(||\mathcal{S}%
_{a,b}^{\left( -,t,t^{\prime }\right) }||)
\end{equation}%
where $||\mathcal{S}_{a,b}^{\left( -,t,t^{\prime }\right) }||$ is the $(%
\mathsf{N}+1)\times (\mathsf{N}+1)$ matrix:%
\begin{eqnarray}
\mathcal{S}_{a,b}^{\left( -,t,t^{\prime }\right) } &\equiv &\Phi
_{a,b+1/2}^{\left( t,t^{\prime }\right) }\text{ \ for \ }a\in \{1,...,%
\mathsf{N}\}, \\
\mathcal{S}_{\mathsf{N}+1,b}^{\left( -,t,t^{\prime }\right) } &\equiv
&\left( \eta _{n}\right) ^{2(b-1)-\mathsf{N}}.
\end{eqnarray}
\end{theorem}

\begin{proof}
We can compute the action of $\sigma _{n}^{-}$, by using the following
reconstruction:%
\begin{equation}
\sigma _{n}^{-}=\prod_{b=1}^{n-1}\mathsf{\bar{T}}(\eta _{b})\mathsf{D}(\eta
_{n})\prod_{b=1}^{n}\frac{\mathsf{\bar{T}}(\eta _{b}/q)}{\det \mathsf{\bar{M}%
}(\eta _{b})},
\end{equation}%
so it holds:%
\begin{equation}
\langle t|\sigma _{n}^{-}|t^{\prime }\rangle =\frac{\prod_{h=1}^{n-1}t(\eta
_{h})\prod_{h=1}^{n}t^{\prime }(\eta _{h}/q)}{\prod_{h=1}^{n}a(\eta
_{h})d(\eta _{h}/q)}\langle t|\mathsf{D}(\eta _{n})|t^{\prime }\rangle .
\end{equation}%
Now from the right $\mathsf{D}$-SOV representation, we have:%
\begin{eqnarray}
\mathsf{D}(\eta _{n})|t^{\prime }\rangle  &=&\sum_{\substack{ h_{1},...,h_{%
\mathsf{N}}=0 \\ \overbrace{{\small h_{n}\text{ is missing}}}}}^{1}\frac{Q_{t^{\prime
}}(\eta _{n}/q)}{\omega(\eta _{n}/q)}(q-\frac{1}{q})\prod_{a\neq
n,a=1}^{\mathsf{N}}\left[ (\frac{\eta _{n}q^{h_{a}}}{\eta _{a}}-\frac{\eta
_{a}}{\eta _{n}q^{h_{a}}})\frac{Q_{t^{\prime }}(\eta _{a}q^{-h_{a}})}{\omega(\eta _{a}q^{-h_{a}})}\right]   \notag \\
&&\times \prod_{1\leq b<a\leq \mathsf{N}}( \eta _{a}^2q^{-2h_{a}}-\eta _{b}^{2}q^{-2h_{b}})|h_{1},...,h_{n}\left. =\right.
1,...,h_{\mathsf{N}}\rangle ,
\end{eqnarray}%
and we can rewrite the coefficient as:%
\begin{equation}
\frac{Q_{t^{\prime }}(\eta _{n}/q)\,q}{\omega(\eta _{n}/q)\eta _{n}^{\mathsf{N+1}}}\prod_{a\neq n,a=1}^{\mathsf{N}}\frac{Q_{t^{\prime }}(\eta
_{a}q^{-h_{a}})q^{h_{a}}}{\omega(\eta _{a}q^{-h_{a}})\eta _{a}}V(\eta
_{1}^{2}/q^{2h_{1}},....,\eta _{n}^{2}/q^{2},....,\eta _{\mathsf{N}%
}^{2}/q^{2h_{\mathsf{N}}},\eta _{n}^{2}),
\end{equation}%
where $V()$ is the determinant of the $(\mathsf{N}+1)\times (\mathsf{N}+1)$
Vandermonde matrix in $\eta _{1}^{2}/q^{2h_{1}},....,\eta
_{n}^{2}/q^{2},....,\eta _{\mathsf{N}}^{2}/q^{2h_{\mathsf{N}}},\eta _{n}^{2}$%
. Now resumming as for the scalar product formula we derive our result. We
have just to notice that the determinant which has in the line $n$ the $\eta
_{n}^{2(b-1)}$ gives zero and so we can add it to derive our result.
\end{proof}

\begin{theorem}
\label{FF-Prop2}Let $\langle t|$ and $|t^{\prime }\rangle $ be a left and a right eigenstate of the transfer matrix $\mathsf{\bar{T}}(\lambda )$, respectively, then it
holds:%
\begin{equation}
\langle t|\sigma _{n}^{z}|t^{\prime }\rangle =\frac{\prod_{h=1}^{n-1}t(\eta
_{h})\prod_{h=1}^{n}t^{\prime }(\eta _{h}/q)}{-\prod_{h=1}^{n}a(\eta
_{h})d(\eta _{h}/q)}\det_{\mathsf{N}+1}(||\mathcal{S}_{a,b}^{\left(
z,t,t^{\prime }\right) }||)
\end{equation}%
where $||\mathcal{S}_{a,b}^{\left( z,t,t^{\prime }\right) }||$ is the $(%
\mathsf{N}+1)\times (\mathsf{N}+1)$ matrix:%
\begin{eqnarray}
\mathcal{S}_{a,b}^{\left( z,t,t^{\prime }\right) } &\equiv &\Phi
_{a,b}^{\left( t,t^{\prime }\right) }\text{ \ for \ }a\in \{1,...,\mathsf{N}%
\},\text{ \ \ \ }b\in \{1,...,\mathsf{N}\} \\
\mathcal{S}_{\mathsf{N}+1,b}^{\left( z,t,t^{\prime }\right) } &\equiv &\eta
_{n}^{2(b-1/2)-\mathsf{N}}\text{ \ \ \ for }b\in \{1,...,\mathsf{N}\} \\
\mathcal{S}_{a,\mathsf{N}+1}^{\left( z,t,t^{\prime }\right) } &\equiv &\frac{%
Q_{t^{\prime }}(\eta _{a}/q)\bar{Q}_{t}(\eta _{a})}{\omega(\eta _{a}/q)}%
\left( \frac{\eta _{a}}{q}\right) ^{\mathsf{N}-1}d(\eta _{a}/q)\text{ \ for
\ }a\in \{1,...,\mathsf{N}\}, \\
\mathcal{S}_{\mathsf{N}+1,\mathsf{N}+1}^{\left( z,t,t^{\prime }\right) }
&\equiv &1/2.
\end{eqnarray}
\end{theorem}

\begin{proof}
We can compute the action of $\sigma _{n}^{z}$, by using the following
reconstruction:%
\begin{equation}
\sigma _{n}^{z}=\prod_{b=1}^{n-1}\mathsf{\bar{T}}(\eta _{b})\frac{\mathsf{C}%
(\eta _{n})-\mathsf{B}(\eta _{n})}{2}\prod_{b=1}^{n}\frac{\mathsf{\bar{T}}%
(\eta _{b}/q)}{\det \mathsf{\bar{M}}(\eta _{b})},
\end{equation}%
so it holds:%
\begin{equation}
\langle t|\sigma _{n}^{z}|t^{\prime }\rangle =\frac{\prod_{h=1}^{n-1}t(\eta
_{h})\prod_{h=1}^{n}t^{\prime }(\eta _{h}/q)}{\prod_{h=1}^{n}a(\eta
_{h})d(\eta _{h}/q)}\langle t|\mathsf{C}(\eta _{n})|t^{\prime }\rangle -%
\frac{\langle t|t^{\prime }\rangle }{2}.
\end{equation}%
Now from the right $\mathsf{D}$-SOV representation of $\mathsf{C}(\eta _{n})$, we
have:%
\begin{eqnarray}
\mathsf{C}(\eta _{n})|t^{\prime }\rangle  &=&\sum_{a=1}^{\mathsf{N}}\sum
_{\substack{ h_{1},...,h_{\mathsf{N}}=0 \\ \overbrace{h_{n}\text{ is missing}%
}}}^{1}\left\{ \prod_{b\neq a,b=1}^{\mathsf{N}}\left[ \frac{(\eta
_{n}^{2}-\eta _{b}^{2}q^{-2h_{b}})}{(\eta _{a}^{2}q^{-2h_{a}}-\eta
_{b}^{2}q^{-2h_{b}})}\frac{Q_{t^{\prime }}(\eta _{b}q^{-h_{b}})}{\omega
_{b}(\eta _{b}q^{-h_{b}})}\right] \right.   \notag \\
&&\times \left. \prod_{1\leq b<a\leq \mathsf{N}}(\eta^{2}_{a}q^{-2h_{a}}-\eta^{2}_{b}q^{-2h_{b}})\frac{\eta
_{a}^{\mathsf{N}-1}Q_{t^{\prime }}(\eta _{a}q^{-h_{a}})d(\eta _{a}q^{-h_{a}})%
}{\eta _{n}^{\mathsf{N}-1}q^{(\mathsf{N}-1)h_{a}}\omega(\eta
_{a}q^{-h_{a}})}\right\} _{h_{n}=1}  \notag \\
&&\times |h_{1},...,h_{n}\left. =\right. 1-\delta _{a,n},...,h_{a}-1,...,h_{%
\mathsf{N}}\rangle ,
\end{eqnarray}%
and so we can write:%
\begin{eqnarray}
\langle t|\mathsf{C}(\eta _{n})|t^{\prime }\rangle  &=&\sum_{a=1}^{\mathsf{N}%
}(-1)^{\mathsf{N}+a}\sum_{\substack{ h_{1},...,h_{\mathsf{N}}=0 \\ 
\overbrace{h_{n}\text{ is missing}}}}^{1}\left\{ \widehat{V}_{a}(\eta
_{1}^{2}/q^{2h_{1}},....,\eta _{\mathsf{N}}^{2}/q^{2h_{\mathsf{N}}},\eta
_{n}^{2})\right.   \notag \\
&&\times \left. \prod_{b\neq a,b=1}^{\mathsf{N}}\frac{Q_{t^{\prime }}(\eta
_{b}q^{-h_{b}})\bar{Q}_{t}(\eta _{b}q^{-h_{b}})}{\omega(\eta
_{b}q^{-h_{b}})}\frac{\eta _{a}^{\mathsf{N}-1}Q_{t^{\prime }}(\eta
_{a}q^{-h_{a}})\bar{Q}_{t}(\eta _{a}q^{1-h_{a}})d(\eta _{a}q^{-h_{a}})}{\eta
_{n}^{\mathsf{N}-1}q^{(\mathsf{N}-1)h_{a}}\omega(\eta _{a}q^{-h_{a}})}%
\right\} _{h_{n}=1}.
\end{eqnarray}%
In the last sum we can reintroduce the sum over $h_{n}$; indeed, $h_{n}=0$
gives zero. Now, it is trivial to remark that the above sum minus $\langle t|t^{\prime }\rangle/2$ is the develop
of the determinant of the $(\mathsf{N}+1)\times (\mathsf{N}+1)$ matrix $||%
\mathcal{S}_{a,b}^{\left( z,t,t^{\prime }\right) }||$ presented in the
statement of proposition.
\end{proof}

\section{Conclusion and outlooks}

\subsection{Results and first prospectives}

In this article we have considered the spin 1/2 highest weight
representations for the 6-vertex Yang-Baxter algebra on a generic $\mathsf{N}$-sites finite lattice and analyzed the integrable quantum models associated
to the antiperiodic transfer matrix. For this integrable quantum models,
which in the homogeneous limit reproduce the $XXZ$ spin 1/2 quantum chain with
antiperiodic boundary conditions, we have obtained the following results:

\begin{itemize}
\item Complete characterization of the transfer matrix spectrum
(eigenvalues/eigenstates) by separation of variables and proof of its
simplicity.

\item Reconstruction of all local operators in terms of the standard
Sklyanin's quantum separate variables.

\item Scalar Products: One determinant formulae of $\mathsf{N}\times $%
$\mathsf{N}$ matrices whose matrix elements are sums over the spectrum of each
quantum separate variable of the product of the coefficients of states; for
all the left/right separate states in the SOV-basis.

\item Form factors of the local spin operators on the transfer matrix
eigenstates in determinant form.
\end{itemize}

In the papers \cite{GN12-2} and \cite{GN12-3},
 the list of fundamental matrix elements for the antiperiodic $XXZ$ spin-1/2 chain is completed with the the computation of the matrix elements on the transfer matrix eigenstates of the so-called density matrix and the two point functions.

Let us complete this subsection mentioning some facts to point out the
relevance of the current results about form factors of local operators.
First of all by using the decomposition of the identity \rf{Id-decomp} we can write any correlation function in spectral
series of form factors. Then it is natural to expect that the correlation functions can be analyzed
numerically mainly by the same tools developed in \cite{CM05} in the ABA
framework and used in the series of works\footnote{It is worth mentioning that important physical observables like the
dynamical structure factors, accessible by neutron scattering experiments 
\cite{Bloch36}-\cite{Balescu75}, were evaluated by this numerical approach.%
} \cite{CM05}-\cite{CCS07}. Indeed, also in our SOV framework we have
determinant representations of the form factors and complete
characterization of the transfer matrix spectrum in terms of the solutions
of a system of Bethe equations. Finally, let us mention the important progresses\footnote{%
Results on asymptotic behaviour which has been also compared with previous
more technical achievements rely mainly on the Riemann-Hilbert analysis of
related Fredholm determinants \cite{KKMST09++}-\cite{K1011}.} achieved
recently \cite{KKMST09}-\cite{KP12} in computing the asymptotic
behavior of correlation functions which are in principle susceptible to be
extended to any (integrable) quantum model possessing determinant
representations for the form factors of local operators \cite{KKMST1110} and
so also to the models analyzed by our approach in the SOV framework.

\subsection{Comparison with other SOV-type results}

In the literature of quantum integrable models there exist results on the
matrix elements of local operators which can be traced back to some
applications of separation of variable methods. In this section, we try to
recall those that we consider more relevant for us also as they allow an
explicit comparison with our results leading to show an universal picture
emerging in the characterization of matrix elements by SOV-methods.

In the case of the quantum integrable Toda chain \cite{Sk1}, Smirnov \cite%
{Sm98} has derived in Sklyanin's SOV framework determinant formulae for
the matrix elements of a conjectured basis of local operators which look
very similar to our formulae. The main difference is due to
the different nature of the spectrum of the quantum separate variables in
these two models. In fact, in the case of the lattice Toda model Sklyanin's measure is continue (continuum SOV-spectrum) while it is discrete
in our case. The elements of the matrices whose determinants give the form
factor formulae are then expressed as \textquotedblleft
convolutions\textquotedblright , over the spectrum of the separate
variables,\ of Baxter equations solutions plus contributions coming from the
local operators. In the case of the Smirnov's formulae they are true
integral being the SOV-spectrum continuum while in our formulae they are
\textquotedblleft discrete convolutions\textquotedblright\ being the
SOV-spectrum discrete. Let us comment that the need to conjecture\footnote{%
The consistence of this conjecture is there verified by a counting arguments
based on the existence of an appropriate set of null conditions for the
\textquotedblleft integral convolutions\textquotedblright .} the form of a
basis of local operators in \cite{Sm98} is due to the lack of a direct
reconstruction of local operators in terms of Sklyanin's separate
variables\footnote{%
It is then worth citing that simple form of reconstructions of local
operators in the quantum Toda model have been achieved by Babelon in \cite%
{OB-04} in terms of a set of quantum separate variables defined by a change
of variables in terms of the original Sklyanin's quantum variables.}.

Even if in the different methodological contest of the S-matrix formulation
of IQFTs, it is worth mentioning that the form factors of local operators 
\cite{Sm92} of the infinite volume quantum sine-Gordon field theory have
also a form similar to the one predicted by SOV.\ This similarity statement
can be made explicit considering for example the \textsf{n}-soliton form
factors for the chiral local operators in the restricted sine-Gordon at the
reflectionless points, $\beta ^{2}=1/(1+\nu )$ with $\nu \in \mathbb{Z}%
^{\geq 0}$; see formula (31) of \cite{BBS96}. Then, for some local fields
(interpreted as primary operators in \cite{BBS96}) the corresponding form
factors can be easily rewritten as determinants of \textsf{n}$\times$\textsf{%
n} matrices whose elements are integral convolution of \textsf{n}-soliton
wave functions (see the $\psi $-functions (32) of \cite{BBS96}) plus
contributions coming from the local operators. The connection of these
results with SOV emerges somehow naturally in \cite{BBS96} as the form
factors of the quantum theory are there identified semi-classically by using
the classical SOV reconstruction of the local fields \cite{BBS96}.

Finally, about representation of form factors in determinant form, it is
worth pointing out the important achievements obtained recently in the
series of works \cite{BJMST07-02}-\cite{JMS11-03} where a fermionic basis of
quasi-local operators has been introduced. There the expectation values of
products of these operators on an appropriate vacuum state are written in
determinant forms, similarly to the free fermions case by Wick's theorem. In
particular, in \cite{JMS11-03} these results on form factors has been
connected to the form factors analysis made in the S-matrix formulation for
the restricted sine-Gordon at the reflectionless points made in \cite{BBS96,BBS97} where as above reported a link to SOV was made.

\subsection{Outlook}

Here, we want to point out the potential generality of the approach
introduced to compute matrix elements of local operators for quantum
integrable models. The main ingredients used to develop it are the solution
of the transfer matrix spectral problem by SOV
construction, the reconstruction of local operators in terms of the quantum
separate variables and the scalar product formulae for the transfer matrix
eigenstates (and general separate states). Then this SOV reconstruction
(inverse problem solution) allows to write the action of any local operator
on transfer matrix eigenstates as a finite sum of separate states in the
SOV-basis. So that the matrix elements of any local operator are written as a
finite sum of determinants of the scalar product formulae. The emerging
picture is the possibility to apply this method to a whole class of
integrable quantum models which were not tractable with other methods and in
particular by algebraic Bethe ansatz.

The main aim is to implement explicitly this approach for a set of key integrable quantum models providing determinant representations
for the matrix elements of local operators. To achieve this
goal is very important as on the one hand it leads to the solution of fundamental quantum models previously unsolved and on the other hand gives the
possibility to develop the mathematical tools to face the same problem for
more involved integrable quantum models. This program has been already
realized for several integrable quantum models as we will summarize in the
following. In the case of the cyclic representations of the 6-vertex
Yang-Baxter algebra like the lattice sine-Gordon model and the $\tau _{2}$%
-model (of special interest for the connection with chiral Potts model) the
form factors of local operators have been derived in \cite{GMN12-SG} and \cite%
{GMN12-T2}, respectively. In \cite{N12-1} this approach is developed for the
higher spin representations of highest weight type of the rational 6-vertex
Yang-Baxter algebra which in the homogeneous limit leads to the higher spin
$XXX$ antiperiodic quantum chain. There the SOV setup is implemented and the
form factors of the local spin operators on the transfer matrix
eigenstates are obtained in a determinant form. In \cite{N12-2} this
approach is developed for the spin 1/2 representations of highest weight
type of the reflection algebra \cite{Skly88}-\cite{GZ94} which in the
homogeneous limit leads to the $XXZ$ open spin 1/2 quantum chains in quite
general non-diagonal boundary conditions. There the SOV setup is implemented
and the matrix elements of some interesting quasi-local string of local
operators are computed. Further matrix elements are computed in \cite{N12-3}
and in \cite{N12-4} for the most general representations of rational type;
the relevance of these findings in the framework of the non-equilibrium
systems like the simple exclusion processes is also discussed there. In \cite%
{N12-5} the SOV setup of the spectral problem is implemented for the spin
1/2 representations of highest weight type of the dynamical 6-vertex
Yang-Baxter algebra and consequently for the corresponding 8-vertex
Yang-Baxter algebra representations. There moreover the scalar product
formulae for these representations are derived in a determinant form which
allows to derive determinant formulae for the form factors as it will be presented in \cite{?NT12}.

\subsubsection{Toward solution of quantum field theories by  integrable microscopic formulation}

We are interested also in the use of integrable quantum models as a tool for the exact and complete characterization of the spectrum and dynamics of quantum field theories (QFTs) going through integrable lattice discretization. Then, it is worth recalling that in the light-cone approach \cite{DdeV87}-\cite{DdeV93} the transfer matrix of the periodic $XXZ$ spin-1/2 quantum chain with alternating inhomogeneities allows to define an integrable lattice regularization of the massive Thirring model. The continuum limit and the infrared limit defining the massive Thirring QFT in the finite volume and the infinite volume, respectively, have been implemented and their spectrum analyzed in the series of papers \cite{DdeV92}-\cite{Rav01}. Then in the infinite volume limit, where the boundary conditions of the original spin chain should not play a role, we have the possibility to describe the massive Thirring QFT by using as starting point the SOV approach of the present paper which on the lattice is known to give a complete characterization of the spectrum. The analysis of this interesting issue by the implementation of the required limits and the comparison with the known results  \cite{DdeV92}-\cite{Rav01}, derived instead in the framework of the algebraic Bethe ansatz, will be the subject of a forthcoming publication.

Finally, let us point out that our interest in integrable lattice regularizations of QFTs is due to the possibility to define an exact setup where to use the reconstruction of local fields in terms of the
quantum separate variables to overcome the longstanding problem of their
identifications in the S-matrix formulation\footnote{%
See \cite{A.Zam77}-\cite{M92} for a review and references
therein.}. Different methods have been introduced to address this problem
and one important line of research is related to the description of massive  integrable quantum field theories (IQFTs) as (superrenormalizable) perturbations of conformal field theories 
\cite{Vi70}-\cite{DFMS97} by relevant local fields \cite{Zam88}-\cite{GM96}.
This characterization has led to the expectation that the perturbations do
not change the structure of the local fields in this way leading to the
attempt to classify the local field content of massive theories\footnote{%
In the S-matrix formulation the local fields are characterized in terms of
form factors (matrix elements on the asymptotic particle states) and many
results are known on these form factors in IQFTs; see for example \cite{KW78}%
-\cite{D04} and references therein.} by that of the corresponding
ultraviolet conformal field theories. Several results are known which
confirm this characterization; see for example \cite{CM90}-\cite{JMT03} for the proof of the isomorphism restricted to the
chiral sector of some IQFTs\footnote{An important role in these studies has been played by the fermionic
representations of the characters, as derived for different classes of
rational conformal field theories in \cite{KKMcM93-1}-\cite{BMc98}.} and the series of works \cite{DN05-1}-\cite{DN08}
where the first rigorous proof of the isomorphism for the full operator
space was given for the massive Lee-Yang model. However, while these are
important results on the global structure of the operator space in the
S-matrix formulation of the massive IQFTs they do not really lead to the
identification of particular local fields\footnote{%
A part for some local fields, like the components of the stress energy
tensor, which can be characterized by physical prescription \cite{DN05-1}
and \cite{DN06}.} which then remain the main missing information in the
S-matrix formulation.

\bigskip
{\bf Acknowledgments}\, G. N. would like to thank N. Grosjean, N. Kitanine, K. K. Kozlowski, B. M.  McCoy, E. Sklyanin, V. Terras, P. Zinn-Justin for their interest in this work and J. M. Maillet for his interest and attentive reading of this paper. G. N. is supported by National Science Foundation grants PHY-0969739. G. N. gratefully acknowledge the YITP Institute of Stony Brook for the opportunity to develop his research programs and the privilege to have stimulating discussions on the present paper and on related subjects with B. M. McCoy. G. N. would like to thank the Theoretical Physics Group of the Laboratory of Physics at ENS-Lyon and the Mathematical Physics Group at IMB of the Dijon University for their hospitality (under support of the grant ANR-10-BLAN-0120-04-DIADEMS) which made possible many fundamental discussions on this and related subjects with N. Grosjean, N. Kitanine, J. M. Maillet and V. Terras.

\section{Appendix}
In this appendix we report the SOV-representations of the Yang-Baxter generators in the right and left eigenbasis of $\mathsf{A}(\lambda )$.

\begin{theorem}
\textsf{I)} \underline{Left $\mathsf{A}(\lambda )$ SOV-representations:} \
Under the condition $\left( \ref{E-SOV}\right) $, the states:%
\begin{equation}
\langle h_{1},...,h_{\mathsf{N}}||\equiv \frac{1}{\text{\textsc{n}}}\langle
0|\prod_{n=1}^{\mathsf{N}}\left( \frac{\mathsf{C}(\eta _{n}/q)}{d(\eta
_{n}/q)}\right) ^{h_{n}},
\end{equation}%
where\ $h_{n}\in \{0,1\},$ $n\in \{1,...,\mathsf{N}\}$, define a $\mathsf{A}%
(\lambda )$-eigenbasis of $\mathcal{L}_{\mathsf{N}}$:%
\begin{equation}
\langle h_{1},...,h_{\mathsf{N}}||\mathsf{A}(\lambda )=a_{\text{\textbf{h}}%
}(\lambda )\langle h_{1},...,h_{\mathsf{N}}||,
\end{equation}%
where:%
\begin{equation}
d_{\text{\textbf{h}}}(\lambda )\equiv \prod_{n=1}^{\mathsf{N}}(\frac{\lambda
q^{(1-h_{n})}}{\eta _{n}}-\frac{\eta _{n}}{\lambda q^{(1-h_{n})}})\text{ \ \
\ and \ \ \textbf{h}}\equiv (h_{1},...,h_{\mathsf{N}}).
\end{equation}%
The action of the remaining Yang-Baxter generators on the generic state $%
\langle h_{1},...,h_{\mathsf{N}}||$ reads:%
\begin{eqnarray}
\langle h_{1},...,h_{\mathsf{N}}||\mathsf{C}(\lambda ) &=&\sum_{a=1}^{\mathsf{N}}\prod_{b\neq a}\frac{\lambda
q^{(1-h_{b})}/\eta _{b}-\eta _{b}/q^{(1-h_{b})}\lambda }{\eta
_{a}q^{(h_{a}-h_{b})}/\eta _{b}-\eta _{b}/q^{(h_{a}-h_{b})}\eta _{a}}d(\eta
_{a}q^{h_{a}-1})\langle h_{1},...,h_{\mathsf{N}}||\text{T}_{a}^{+}, \\
&&  \notag \\
\langle h_{1},...,h_{\mathsf{N}}||\mathsf{B}(\lambda ) &=&\sum_{a=1}^{\mathsf{N}}\prod_{b\neq a}\frac{\lambda
q^{(1-h_{b})}/\eta _{b}-\eta _{b}/q^{(1-h_{b})}\lambda }{\eta
_{a}q^{(h_{a}-h_{b})}/\eta _{b}-\eta _{b}/q^{(h_{a}-h_{b})}\eta _{a}}a(\eta
_{a}q^{h_{a}-1})\langle h_{1},...,h_{\mathsf{N}}||\text{T}_{a}^{-}.
\end{eqnarray}%
where:%
\begin{equation}
\langle h_{1},...,h_{a},...,h_{\mathsf{N}}||\text{T}_{a}^{\pm }=\langle
h_{1},...,h_{a}\pm 1,...,h_{\mathsf{N}}||.
\end{equation}%
Finally, $\mathsf{D}(\lambda )$ is uniquely defined by the quantum
determinant relation.\smallskip

\textsf{II)} \underline{Right $\mathsf{A}(\lambda )$ SOV-representations:} \
Under the condition $\left( \ref{E-SOV}\right) $, the states:%
\begin{equation}
||h_{1},...,h_{\mathsf{N}}\rangle \equiv \frac{1}{\text{\textsc{n}}}%
\prod_{n=1}^{\mathsf{N}}\left( \frac{\mathsf{B}(\eta _{n}/q)}{a(\eta _{n})}%
\right) ^{h_{n}}|0\rangle ,
\end{equation}%
where\ $h_{n}\in \{0,1\},$ $n\in \{1,...,\mathsf{N}\}$, define a $\mathsf{A}%
(\lambda )$-eigenbasis of $\mathcal{R}_{\mathsf{N}}$:%
\begin{equation}
\mathsf{A}(\lambda )||h_{1},...,h_{\mathsf{N}}\rangle =a_{\text{\textbf{h}}%
}(\lambda )|h_{1},...,h_{\mathsf{N}}\rangle .
\end{equation}%
The action of the remaining Yang-Baxter generators on the generic state $%
||h_{1},...,h_{\mathsf{N}}\rangle $ reads:%
\begin{eqnarray}
\mathsf{C}(\lambda )||h_{1},...,h_{\mathsf{N}}\rangle &=&\sum_{a=1}^{\mathsf{N}}\text{T}_{a}^{-}||h_{1},...,h_{\mathsf{N}}\rangle\prod_{b\neq a}\frac{\lambda
q^{(1-h_{b})}/\eta _{b}-\eta _{b}/q^{(1-h_{b})}\lambda }{\eta
_{a}q^{(h_{a}-h_{b})}/\eta _{b}-\eta _{b}/q^{(h_{a}-h_{b})}\eta _{a}}d(\eta
_{a}q^{-h_{a}}), \\
&&  \notag \\
\mathsf{B}(\lambda )||h_{1},...,h_{\mathsf{N}}\rangle &=&\sum_{a=1}^{\mathsf{N}}\text{T}_{a}^{+}||h_{1},...,h_{\mathsf{N}}\rangle\prod_{b\neq a}\frac{\lambda
q^{(1-h_{b})}/\eta _{b}-\eta _{b}/q^{(1-h_{b})}\lambda }{\eta
_{a}q^{(h_{a}-h_{b})}/\eta _{b}-\eta _{b}/q^{(h_{a}-h_{b})}\eta _{a}}a(\eta
_{a}q^{-h_{a}}).
\end{eqnarray}%
where:%
\begin{equation}
\text{T}_{a}^{\pm }||h_{1},...,h_{a},...,h_{\mathsf{N}}\rangle
=||h_{1},...,h_{a}\pm 1,...,h_{\mathsf{N}}\rangle .
\end{equation}%
Finally, $\mathsf{D}(\lambda )$ is uniquely defined by the quantum
determinant relation.
\end{theorem}
Note that the proof of this theorem can be given along the same lines used to prove the Theorem \ref{Th-D-SOV}.
\begin{small}

\end{small}

\end{document}